\documentclass[aps,a4paper,11pt]{article}
\usepackage[a4paper]{geometry}
\newgeometry{margin=2.5cm}

\usepackage[english]{babel}
\usepackage[latin1]{inputenc}       
\usepackage[T1]{fontenc}

\usepackage{amssymb}
\usepackage{amsmath}
\usepackage{amsthm}
\usepackage{mathabx}
\usepackage{mathrsfs}
\usepackage{authblk}
\usepackage{comment}

\usepackage[normalem]{ulem}

\newcommand{\ee}{{\mathrm e}}
\newcommand{\ii}{{\mathrm i}}
\newcommand{\dd}{{\mathrm d}}

\newcommand{\CC}{\mathbb{C}}

\newcommand{\RR}{\mathbb{R}}

\newcommand{\ZZ}{\mathbb{Z}}

\newcommand{\wpsi}{\widehat \psi}

\usepackage{graphicx}
\usepackage{xcolor}

\definecolor{blue_M}{rgb}{0.37,0.51,0.71}
\definecolor{orange_M}{rgb}{0.88,0.61,0.14}

\usepackage{tikz}
\usetikzlibrary{arrows}
\usetikzlibrary{shapes.geometric}
\usetikzlibrary{decorations.markings}
\tikzstyle{arrowmid}[0.5]=[decoration=
{markings,
mark=at position #1 with {\arrow{>}}},
postaction={decorate}]

\theoremstyle{plain}
\newtheorem{thm}{Theorem}[section]
\newtheorem{cor}[thm]{Corollary}
\newtheorem{lem}[thm]{Lemma}
\newtheorem{prop}[thm]{Proposition}
\newtheorem{propdef}[thm]{Proposition/Definition}

\theoremstyle{definition}
\newtheorem{defn}[thm]{Definition}
\newtheorem{rem}[thm]{Remark}

\newtheorem{exmp}[thm]{Example}

\title{Topology in shallow-water waves: \\ A spectral flow perspective}
\author[1]{Clément Tauber\thanks{clement.tauber@math.unistra.fr}}
\author[2]{Guo Chuan Thiang\thanks{guochuanthiang@bicmr.pku.edu.cn}}
\affil[1]{Institut de Recherche Mathématique Avancée, UMR 7501 Université de Strasbourg et CNRS, 7 rue René-Descartes, 67000 Strasbourg, France}
\affil[2]{Beijing International Center for Mathematical Research, Peking University, No. 5 Yiheyuan Road Haidian District, Beijing, P.R. China 100871}
\date{\today}
\begin{document}
\maketitle

\begin{abstract}
In the context of topological insulators, the shallow-water model was recently shown to exhibit an anomalous bulk-edge correspondence. For the model with a boundary, the parameter space involves both longitudinal momentum and boundary conditions, and exhibits a peculiar singularity. We resolve the anomaly in question by defining a new kind of edge index as the spectral flow around this singularity. Crucially, this edge index samples a whole family of boundary conditions, and we interpret it as a boundary-driven quantized pumping. Our edge index is stable due to the topological nature of spectral flow, and we prove its correspondence with the bulk Chern number index using scattering theory and a relative version of Levinson's theorem. The {full} spectral flow structure of the model is also investigated.
\end{abstract}

\section{Introduction}

Bulk-edge correspondence is a central concept of topological insulators. It states that topological indices defined independently for an infinite sample -- the bulk -- and for a half-infinite sample with boundary  -- the edge -- actually coincide. Originally investigated in Quantum Hall effect \cite{Halperin82}, bulk-edge correspondence theorems have then been established for a wide range of models, both discrete \cite{Hatsugai,SchulzBaldesKellendonkRichter00,GrafPorta,Avilaetal13,ProdanSchulzBaldes16,EssinGurarie11} and continuous \cite{AvronSeilerSimon94,BellissardVanElstSchulzBaldes94,CombesGerminet05,BourneRennie18,Drouot19}. Applications go beyond condensed matter physics, such as in optics \cite{RaghuHaldane08,DeNittisLein17}, acoustics \cite{Perietal19}, or fluid dynamics \cite{DelplaceMarstonVenaille17}, actually to any wave phenomenon that can be described by a self-adjoint operator. 

Recently, the bulk-edge correspondence has been shown to be violated for the first time in the shallow-water model \cite{TauberDelplaceVenaille19bis,GJT21}, that describes oceanic layers on Earth. Once regularized by an odd-viscous term, such a model has a well-defined bulk index: the Chern number. However, the edge picture is anomalous: there exists a family of self-adjoint boundary conditions such that the number of edge modes changes with the choice of boundary condition. The origin of this anomaly is rooted in the unbounded nature of the operator's spectrum and was analyzed via scattering theory and a relative version of Levinson's theorem, originally developed in \cite{GrafPorta}.

The aim of this paper is to circumvent this anomaly and restore a bulk-edge correspondence theorem for shallow-water waves, by defining the edge index in a manner which can sample a whole family of boundary conditions at once. We focus on the same family of boundary conditions studied in \cite{GJT21}, and first show that this provides a continuous parametrization of boundary operators. There is a special point in the joint parameter space of momentum and boundary conditions, at which self-adjointness is lost. We define the edge index as the spectral flow along a loop encircling this special point, and show that it {defines} a  topological index which coincides with the bulk Chern number. The situation is similar to that of Weyl semimetals, which exhibits spectral flow around singularities in the boundary momentum space where the Fredholm condition is lost \cite{Thiang21}.

Spectral flow is a powerful concept that has been used for bulk-edge correspondences in 2D discrete models of topological insulators in \cite{GrafPorta,DeNittisSchulzBaldes16,Braverman}, and 3D Weyl semimetals in \cite{Thiang21,Gomi}. In our context, the spectral flow along a loop counts the signed number of crossings of edge mode branches with an energy level, as the loop is traversed. Usually, a boundary condition is fixed when computing the number of edge modes, but this way of counting leads to an anomalous mismatch with the bulk index \cite{GJT21}. With hindsight, it is unsurprising that the correct edge index should sample different boundary conditions, since the bulk index does not pick out any particular boundary condition. Furthermore, our edge index reveals the intriguing occurrence of quantized pumping over a cycle of boundary conditions.

The interpretation of spectral flow as a topological quantum pump has a longstanding tradition, where the system is usually probed via cycles of external parameters such as fluxes \cite{Laughlin81,NiuThoulessWu85,AvronSeilerSimon94}. To our knowledge, our edge index is the first to involve a cycle of genuine boundary conditions of the sample. Fictitious cycles of boundary conditions were previously considered by the second author in \cite{Thiang21, CareyThiang}. 

Finally, we show that our new edge index is stable against perturbations, a crucial feature of topological quantities. This relies on a generalization of the naive spectral flow definition, based on \cite{P96,BLP05}, which deals with the occurrence of possible pathological edge mode branches. As a byproduct we investigate the complete spectral flow structure of the model.

\medskip

\paragraph{Outline.} In Section \ref{sec:Laplace.pump}, we provide, as a warm-up, an elementary example of the new type of topological quantum pump discussed above. Our main results on the bulk-boundary correspondence for the shallow-wave model are presented in Section~\ref{sec:main}. The proofs are detailed in Sections~\ref{sec:proof_continuity} and \ref{sec:BEC}. Finally, in Section~\ref{sec:sf.structure}, we discuss the full spectral flow structure of the shallow-water model on a half-space.

\section{Charge pumping via a cycle of boundary conditions}\label{sec:Laplace.pump}
On a line, the free Schr\"{o}dinger particle governed by the Hamiltonian $\Delta=-\frac{d^2}{dx^2}$ has spectrum $[0,\infty)$. There appears to be nothing topological about this operator, i.e.\ it has no bulk topological index.

We might therefore expect the same for the particle on the half-line $x\geq 0$. This is not the case. To conserve probability, we need self-adjoint boundary conditions at $x=0$, and these are given by (see pp.\ 144 of \cite{RS2})
\begin{equation*}
\psi'(0)+a\psi(0)=0,\qquad a\in\RR\cup\{\infty\}.
\end{equation*}
Here, $a=0$ labels the Neumann Laplacian on the half-line, while the Dirichlet Laplacian is labelled by $a=\infty$, which is identified with $a=-\infty$. The remaining $a\neq 0,\infty$ label Robin boundary conditions. Altogether, the parameter space of boundary conditions is the compactified line $\RR\cup\{\infty\}$ (i.e.\ a circle), and we write $\Delta_a$ for the half-line Laplacian subject to the $a$-boundary condition.

The essential spectrum of each $\Delta_a$ remains $[0,\infty)$. However, once $a>0$, the operator $\Delta_a$ acquires an eigenvalue $-a^2$ with normalizable eigenfunction $e^{-ax}$. Thus, as the parameter $a$ is increased along the loop $\RR\cup\{\infty\}$, one eigenvalue emerges out of the essential spectrum {when $a=0$, and it decreases towards $-\infty$ as $a$ is increased towards $\infty$.} This ``charge pump'' is illustrated in Figure~\ref{fig:freecase} below.
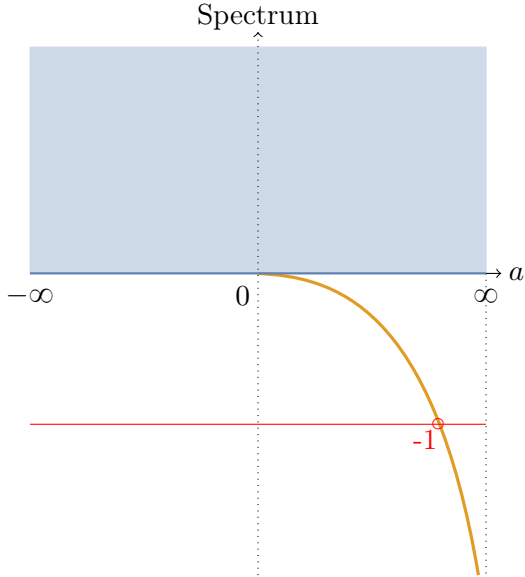
\begin{figure}[ht]
\centering
\begin{tikzpicture}
\filldraw[blue_M!30] (-3,0) rectangle (3,3);
\draw[->] (-3,0) -- (3.2,0);
\draw[dotted,->] (0,-4) -- (0,3.2);
\draw[dotted] (3,0) -- (3,-4);
\node at (3.4,0) {$a$};
\node at (0,3.4) {Spectrum};
\node at (-3,-0.3) {$-\infty$};
\node at (3,-0.3) {$\infty$};
\node at (-0.20,-0.3) {$0$};
\draw[very thick,orange_M] (0,0) .. controls (1,-0.05) and (2.3,-0.4) .. (2.9,-4);
\draw[thick,blue_M] (-3,0)--(3,0);
\draw[red] (-3,-2)--(3,-2);
\node[red] at (2.37,-2) {$\circ$};
\node[red] at (2.2,-2.2) {-1};
\end{tikzpicture}
\caption{The Laplace operator on the half-line {has boundary-driven spectral flow of $-1$ across any negative energy level, {as the boundary condition parameter $a$ is varied around a loop.}}\label{fig:freecase}}
\end{figure}

To achieve the above non-trivial spectral flow, it is important that we are dealing with unbounded operators, so that the negative eigenvalue does not have to rejoin the  {negative} essential spectrum {(which is empty)}. The family $\{\Delta_a+1\}_{a\in\RR\cup\{\infty\}}$ is an unbounded self-adjoint Fredholm loop, for which a precise theory of spectral flow  and its topological invariance was developed in \cite{BLP05,P96}. 

This example indicates that we need to carefully account for the topology of boundary conditions, when investigating bulk-boundary correspondence in continuum models which are spectrally unbounded, such as the shallow-water wave model. There are, however, two main differences with the simple model above. First, the shallow-water model is two-dimensional with translation invariance along the boundary, so that the parameter space is actually a (punctured) cylinder instead of a circle. Second, its essential spectrum is unbounded from above and below, which leads to more subtle situations for eigenvalue branches.

\section{Shallow-water model and main results\label{sec:main}}

\subsection{Bulk model and Chern number}
The linear, rotating and odd-viscous shallow-water model describes a two-dimensional thin layer of fluid of height $\eta$ and velocity $(u,v)$. Such variables are ruled by a system of linear partial differential equations that can be rewritten as a Schr\"odinger equation \cite{TDV19,GJT21} :
\begin{equation}\label{eq:ShallowWater_Schrodinger} 
	\ii \partial_t \psi = \mathcal H \psi, \qquad \psi = \begin{pmatrix}
		\eta \\ u\\ v
	\end{pmatrix}, \qquad \mathcal H = \begin{pmatrix}
		0 & p_x & p_y \\ p_x & 0 & -\ii (f-\nu  p^2) \\ p_y & \ii (f-\nu  p^2) & 0
	\end{pmatrix},
\end{equation}
with $p_x = -\ii \partial_x$, $p_y = -\ii \partial_y$  and ${p}^2 = p_x^2+p_y^2$. Here, $f, \nu$ are model constants satisfying
$f>0$, $\nu >0$ and $1-4f\nu>0$. 
The operator $\mathcal H$ is self-adjoint on the domain $H^1(\mathbb R^2)\oplus H^2(\mathbb R^2) \oplus H^2(\mathbb R^2)$ (denoting the usual Sobolev spaces) in $L^2(\mathbb R^2,\mathbb C^3)$. Due to translation invariance, the solutions of \eqref{eq:ShallowWater_Schrodinger} decompose into Fourier modes $\psi =\widehat\psi \ee^{\ii(k_x x + k_y y -\omega t)}$, with momentum $(k_x,k_y)\in \mathbb R^2$ and frequency $\omega \in \mathbb R$. This reduces \eqref{eq:ShallowWater_Schrodinger} to $H \wpsi = \omega \wpsi$ where $H(k_x,k_y)$ is a family of $3\times 3$ Hermitian matrices. Solving the eigenvalue equation leads to three frequency bands: 
\begin{equation}\label{eq:bulkbands}
\omega_\pm({k_x,k_y}) = \pm \sqrt{k^2+(f-\nu k^2)^2},\qquad \omega_0({k_x,k_y}) =0,
\end{equation}
with $k^2=k_x^2+k_y^2$. The bands are separated by two spectral gaps: $(0,f)$ and $(-f,0)$. It was shown in \cite{TDV19,GJT21} that each band is associated to a well defined topological index, the Chern number, with respective values $C_\pm = \pm2$ and $C_0=0$, attesting to non-trivial topology in the bulk (namely when $\mathcal H$ acts in $L^2(\mathbb R^2,\mathbb C^3)$).

\subsection{Half-space model and its singularity}
One anticipates a corresponding topological invariant of edge states in the half-plane problem. We therefore consider $\mathcal{H}$ of Eq.\ \eqref{eq:ShallowWater_Schrodinger} acting on the upper half-plane $y\geq 0$, denoting it by $\mathcal{H}^\sharp$. As a general convention, we will use a $\sharp$ superscript to denote half-space operators. We analyse translation-invariant (in the $x$-direction) boundary conditions, so that the momentum $k_x$ is still conserved. 

When we Fourier transform $\mathcal{H}$ of Eq.\ \eqref{eq:ShallowWater_Schrodinger} only in  the $x$-direction, we obtain a family $H(k_x),\, k_x\in \RR$ of ordinary differential operators acting on a line, explicitly given by the expression on the right side of Eq.\ \eqref{eqn:formal.1D} below. The spectrum of $H(k_x)$ is {the union over $k_y\in\RR$ of the spectrum of $H(k_x,k_y)$ (found in Eq.\ \eqref{eq:bulkbands} above), namely,}
\begin{equation}
\sigma(H(k_x))=\; (-\infty,-\sqrt{k_x^2+(f-\nu k_x^2)^2}]\,\cup\,\{0\}\,\cup\,[\sqrt{k_x^2+(f-\nu k_x^2)^2},\infty).\label{eq:bulk.spectrum.kx}
\end{equation}
{Note that $0$ is infinitely-degenerate, thus an essential spectral point of $H(k_x)$.}

Let $H^\sharp(k_x)$ denote the same formal differential operator as $H(k_x)$, but acting only on the \emph{half-line} $y\geq 0$:
\begin{equation}
	H^\sharp(k_x)=
	\begin{pmatrix} 0 & k_x & -\ii\frac{d}{dy} \\ k_x & 0 & -\ii(f-\nu(k_x^2-\frac{d^2}{dy^2})) \\ -\ii\frac{d}{dy} & \ii(f-\nu(k_x^2-\frac{d^2}{dy^2})) & 0\end{pmatrix}.
	\label{eqn:formal.1D} 
\end{equation}
Then $H^\sharp(k_x)$ is symmetric (i.e.\ formally self-adjoint) on the initial dense domain $C_c^\infty((0,\infty);\CC^3)\subset L^2((0,\infty);\CC^3)$, and we use the same symbol for its closure. It may be extended to a self-adjoint operator by specifying appropriate boundary conditions at $y=0$. As in \cite{GJT21}, we shall focus on the following family of boundary conditions:
\begin{equation}
	v|_{y=0} = 0, \qquad (\ii k_x u + a \partial_y v)|_{y=0} = 0. \label{eq:boundary_condition}
\end{equation}
{Here, the real parameters $k_x,a$ initially parametrize a plane, but $a=\infty$ is also allowed, and defines the same boundary condition as $a=-\infty$. So $a\in\RR\cup\{\infty\}$ parametrizes a compactified line, i.e.\ a circle, and $(k_x,a)$ parametrizes an infinite cylinder $\RR\times(\RR\cup\{\infty\})$.}

{Crucially, the $k_x=0=a$ case gives a \emph{vacuous} and hence non self-adjoint condition, and must be excluded. Thus the correct parameter space is a \emph{punctured} cylinder,
\begin{equation*}
	(k_x,a) \in (\RR\times (\RR\cup\{\infty\}))\setminus \{(0,0)\}=: \mathring{C},
\end{equation*}
as illustrated in Figure~\ref{fig:cylinder}. In Section \ref{sec:sf.structure}, we will see that this singularity of self-adjointness is the source of the topology of the half-plane model.}

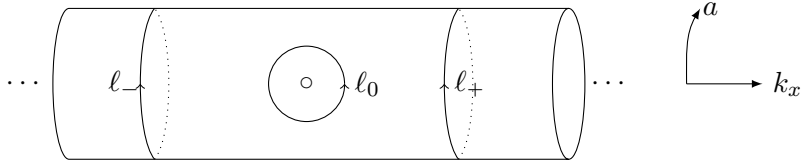
\begin{figure}[ht]
\centering
\begin{tikzpicture}
 
\node[cylinder, draw, minimum width = 2cm,
    minimum height = 7cm]
     (c) at (0,0) {$\circ$};
\node at (-3.7,0) {$\cdots$};
\node at (4,0) {$\cdots$};
\node (r) at (6,0) [right] {$k_x$};
\node (u) at (5.08,1) [right] {$a$};
\draw[-latex] (5,0) to (r);
\draw[-latex](5,0) to [out=90,in=240] (5.18,1);
\draw[->] (0.5,0) arc (0:360:0.5);
\node at (0.8,0) {$\ell_0$};
\draw[arrowmid] (-2,-1) ..  controls (-2.25,-0.7) and (-2.25,0.7) .. (-2,1);
\draw[dotted] (-2,-1) ..  controls (-1.75,-0.7) and (-1.75,0.7) .. (-2,1);
\node at (-2.4,0) {$\ell_-$};
\draw[arrowmid] (2,-1) ..  controls (1.75,-0.7) and (1.75,0.7) .. (2,1);
\draw[dotted] (2,-1) ..  controls (2.25,-0.7) and (2.25,0.7) .. (2,1);
\node at (2.15,0) {$\ell_+$};
\end{tikzpicture}
\caption{The parameter space $\mathring{C}$ of longitudinal momentum $k_x\in\mathbb{R}$ and boundary condition $a\in\mathbb{R}\cup\{\infty\}$. The singularity $(0,0)$ is removed, so that $\mathring{C}$ is a punctured cylinder. The directed loops $\ell_0, \ell_+, \ell_-$ are non-contractible in $\mathring{C}$.}\label{fig:cylinder}
\end{figure}

 The following continuity result is proved in Section \ref{sec:proof_continuity}.  

\begin{thm}\label{thm:gap.continuity}
	For each point $(k_x,a) \in \mathring{C}$, the boundary condition $\eqref{eq:boundary_condition}$ provides a self-adjoint extension of $H^\sharp(k_x)$ that we denote by $H^\sharp(k_x,a)$. Moreover, the resolvent assignment $$(k_x,a)\mapsto \left(H^\sharp(k_x,a)\pm\ii\right)^{-1}$$ is norm-continuous on the punctured cylinder $\mathring{C}$.
\end{thm}

The continuity of the family $H^\sharp(k_x,a)$ is an essential step to prove the stability of the {edge} index defined below. Moreover, a continuity property with respect to a boundary condition is not straightforward and requires von Neumann theory of deficiency indices for self-adjoint extensions. See Section \ref{sec:proof_continuity} below. Finally, because we deal with unbounded operators, it is natural to work with the {topology of norm-resolvent convergence}.

\subsection{Bulk-edge correspondence}   For any fixed $R>0$, consider the following {non-contractible} loop in the parameter space $\mathring{C}$:
\begin{equation}\label{eq:def_CR}
\mathcal C_R = \{(k_x,a) = (R\cos \theta,R \sin \theta), \theta \in [-\pi,\pi]\} \subset \mathring C.
\end{equation}
It is homotopic to $\ell_0$ from Figure~\ref{fig:cylinder}. Each point in $\mathcal C_R$ defines a self-adjoint operator
\begin{equation*}
H^\sharp_{R,\theta}:=H^\sharp(R\cos \theta, R\sin \theta),
\end{equation*}
continuously parametrized by $\theta$ due to Theorem \ref{thm:gap.continuity}. 

{The essential spectrum of $H^\sharp_{R,\theta}$ is obtained by putting $k_x=R\cos \theta$ in Eq.\ \eqref{eq:bulk.spectrum.kx} for the spectrum of the bulk operator $H(k_x)$,} 
\begin{equation}\label{eq:esspec_Hsharp_theta}
	\sigma_\mathrm{ess}(H^\sharp_{R,\theta}) = (-\infty, -\omega_\theta] \cup \{0\} \cup [\omega_\theta, +\infty),\qquad \omega_\theta = \sqrt{R^2\cos(\theta)^2 + (f-\nu R^2\cos(\theta)^2)^2},
\end{equation}
{while} discrete eigenvalues {may appear} inside the gaps $(-\omega_\theta,0)\cup(0,\omega_\theta)$.  As $\theta$ varies, such eigenvalues  {trace out} continuous edge mode branches that live inside the local {essential spectral} gap of the family $\{H_{R,\theta}^\sharp\}_\theta$, see Figure~\ref{fig:naive_spectral_flow}.  We shall focus on {eigenvalue branches in} the upper gap $(0,\omega_\theta)$ from now on.  Recall that the global upper gap is $(0,f)$, so that $(0,f)\subset (0,\omega_\theta)$ for all $\theta$.

\begin{defn}\label{defn:edge.index.naive} 
		For $\mu \in (0,f)$, the number of edge modes $n(\mathcal C_R,\mu)$ is defined as the {signed} number of crossings of {the} edge modes branches of $H^\sharp_{R,\theta}$ with the fiducial line $\omega=\mu$, counted positively (resp. negatively) if the branch has negative slope (resp. positive slope) at the crossing, as $\theta$ increases.
\end{defn}

We illustrate this definition in Figure~\ref{fig:naive_spectral_flow} for various values of $R$. 

\begin{figure}[htb]
	\centering
	\includegraphics[scale=0.45]{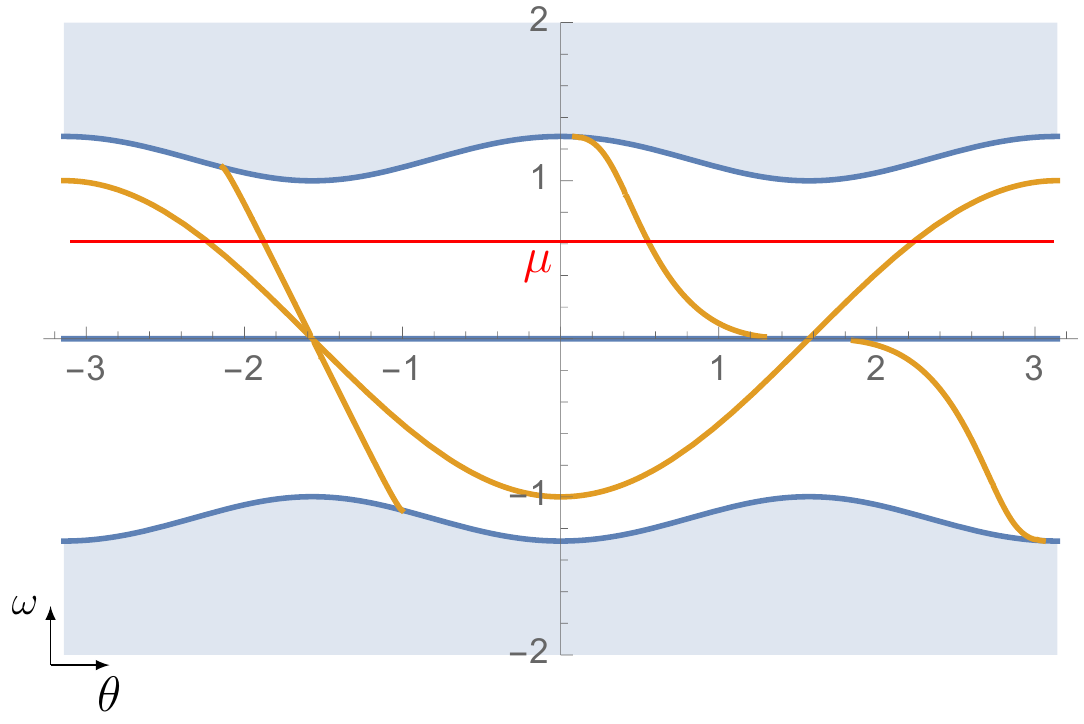} \includegraphics[scale=0.45]{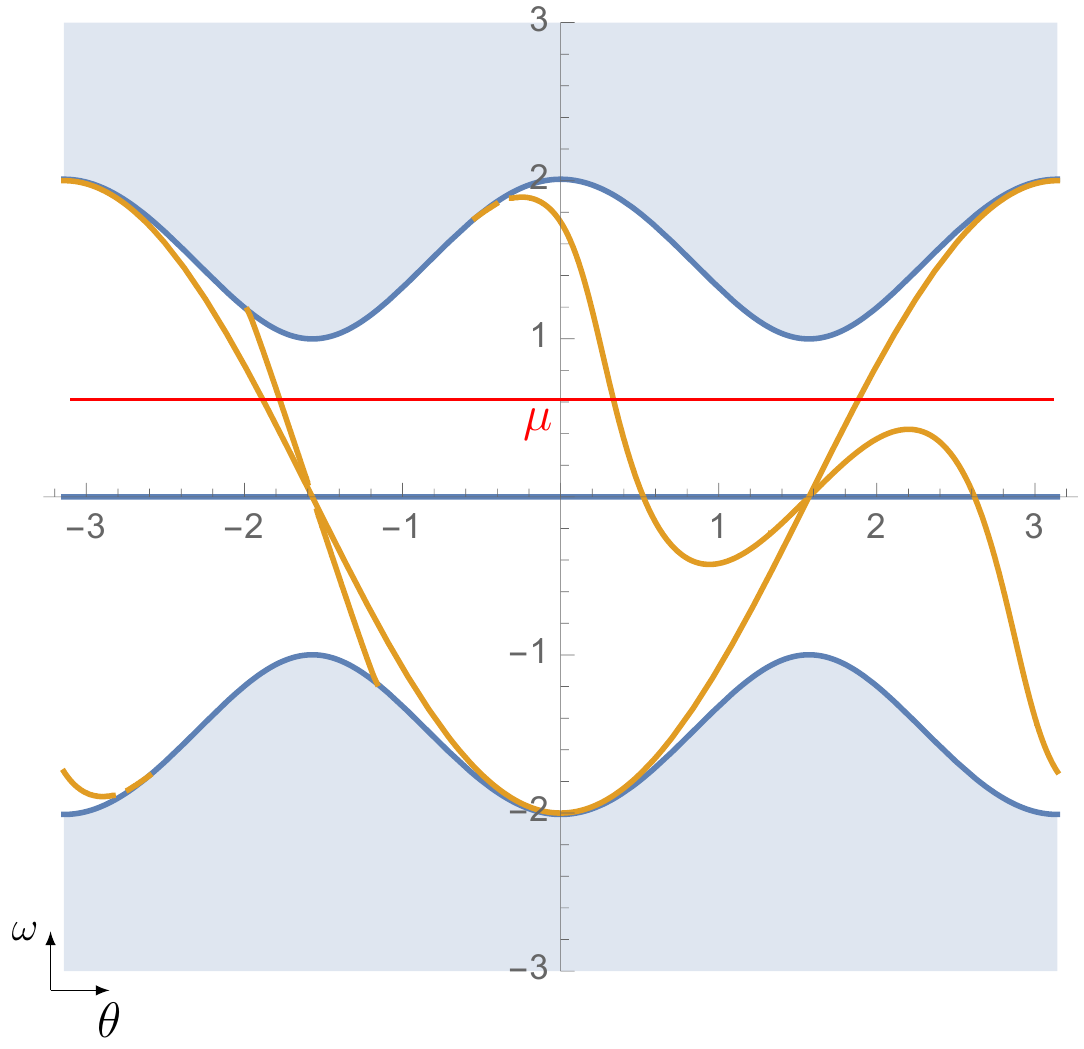} \includegraphics[scale=0.45]{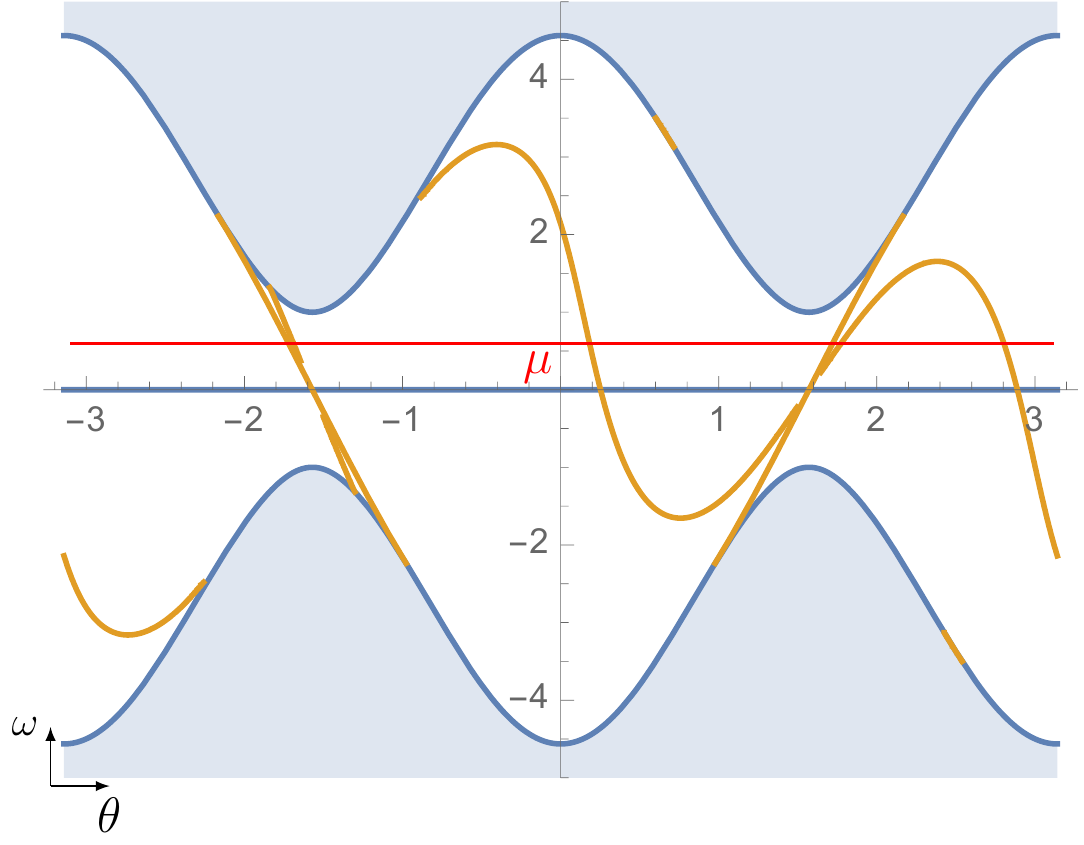}
	\caption{Spectrum of $H_{R,\theta}^\sharp$  for $f=1$ and $\epsilon=0.2$, along  $\mathcal C_R$ with $R=1,\, 2$ and $4$, respectively on the left, middle and right. Shaded blue regions correspond to extended solutions whereas yellow branches are edge modes confined near the boundary. {The number} $n(\mathcal C_R, \mu)$  {of} signed intersections of such branches with the red line $\omega=\mu$
		 is always $2$, {for any choice of $R>0$ and $\mu\in (0,f)$.}
		 \label{fig:naive_spectral_flow}}
\end{figure}

\begin{thm}[Bulk-edge correspondence]\label{thm:naive_BEC}
The number of edge modes $n(\mathcal C_R,\mu)$ is independent of $R>0$ and $\mu\in (0,f)$, and matches with the difference of bulk Chern numbers:
$$
n(\mathcal C_R, \mu) = C_+- C_0.
$$
\end{thm}
In \cite{GJT21} {one attempted to count the number of edge modes} 
for a fixed boundary condition $a\neq0$. {This violated a} more standard formulation of bulk-edge correspondence along the lines of \cite{GrafPorta, Hatsugai}. 
{In fact, the number of edge mode branches crossing $\omega=\mu$ depends on the choice of $\mu\in (0,f)$,} see \cite[Fig. 2]{GJT21}, so that another definition for edge mode counting has to be used {(\cite[Definition 4]{GJT21})}. {Still,} the system remains anomalous: the number of edge modes changes with $a$. The origin of the anomaly is rooted in the infinite region of the spectrum, as $k_x$ and $\omega$ go to infinity. 
{In this paper, we restrict} the parameters to a compact loop $\mathcal C_R$,  {and Theorem \ref{thm:naive_BEC}} above shows that the bulk-edge correspondence can be restored. 
The edge index $n(\mathcal C_R, \mu)$ can be interpreted as a quantized pumping {of edge modes from the upper to the middle band} when both $k_x$ and $a$ are driven together along the loop $\mathcal C_R$, {see Remark~\ref{rk:fixed_kx}.}

The proof of Theorem~\ref{thm:naive_BEC} is done in Section \ref{sec:BEC}. We {extend the approach from} \cite{GJT21} based on scattering amplitudes and Levinson's theorem, originally developed in \cite{GrafPorta} for discrete models. 


\subsection{Spectral flow stability} \label{sec:main_stability}
An essential feature of topological indices is their stability against perturbations. In Theorem \ref{thm:naive_BEC}, the edge mode count via transversal crossings, $n(\mathcal C_R, \mu)$, is well-defined due to regularity properties of the edge modes branches for the unperturbed $H_{R,\theta}$, a family of order 2 ordinary differential operators with constant coefficients. For example, at each $\theta$, there are only finitely many edge modes, so they do not accumulate into the bulk spectrum. Moreover, as $\theta$ is increased, the number of merging points with the bulk spectrum, used in Definition~\ref{defn:edge.index} later, is also finite. 
Such properties may not persist if we perturb $H_{R,\theta}$. 

It is therefore worth pointing out that $n(\mathcal C_R, \mu)$ is, up to a minus sign, a special case of \emph{spectral flow} a la Phillips \cite{P96,BLP05}, whose precise definition and homotopy invariance is recalled in Section \ref{sec:sf.structure}, Definition \ref{defn:spectral.flow}.
Crucially, the spectral flow formulation of edge mode counting still {makes sense} in the perturbed setting, where the aforementioned pathologies {in the edge mode branches} may occur. The only requirement is continuity with respect to the path parameter, which is ensured by Theorem~\ref{thm:gap.continuity}. 

\begin{cor}[{Bulk-edge correspondence, perturbed version}]\label{cor:perturbed}
Let $V=\{V(k_x)\}_{k_x \in \mathbb R}$ be a continuous and uniformly bounded family of self-adjoint operators on $L^2((0,\infty),\mathbb C^3)$. Assume each $V(k_x)$ to be relatively compact with respect to $H^\sharp(k_x)$. Let 
\begin{equation*}
\ell_0:[0,2\pi]\rightarrow \mathring{C},\qquad\theta\mapsto \ell_0(\theta)\equiv(k_x(\theta),a(\theta))
\end{equation*}
be any continuous loop winding around $(k_x,a)=(0,0)$ once in the anticlockwise sense. Then the operator loop
\begin{equation*}
\tilde{\ell_0}:\theta\mapsto H^\sharp(\ell_0(\theta))+V(k_x(\theta))-\mu
\end{equation*}
has well-defined spectral flow {${\rm Sf}(\tilde{\ell}_0)$} {across} any $\mu\in(0,f)$, which equals
$$
{\rm Sf}(\tilde{\ell}_0) = -C_+.
$$
\end{cor}
\begin{proof}
The essential spectrum of the perturbed system $H^\sharp(k_x,a)+V(k_x)-\mu$ does not depend on $V$, and it has a well-defined spectral flow along any loop in $\mathring{C}$ due to Theorem \ref{thm:gap.continuity}. By homotopy invariance of spectral flow, we may continuously turn the perturbation $V$ off and deform $\ell_0$ to the loop $\mathcal{C}_R$ without changing the spectral flow. Thus $ -{\rm Sf}(\tilde{\ell}_0)= -{\rm Sf}(\mathcal{C}_{R})\equiv n(\mathcal C_R,\mu)=C_+$ using Theorem \ref{thm:naive_BEC}.
\end{proof}

\begin{exmp} Suppose the presence of the boundary induces an $x$-independent potential function $g\equiv g(y)$ ($3\times 3$ Hermitian matrix-valued), which is continuous and decays to 0 as $y\to \infty$. The perturbation $V(k_x)$ is simply the multiplication operator $M_g$ for all $k_x$.
\end{exmp}

Section~\ref{sec:sf.structure} is dedicated to an exhaustive study of the spectral flow structure of the family $\{H^\sharp(k_x,a)\}_{(k_x,a)\in\mathring{C}}$. There, we extend the definition of spectral flow across any energy curve lying in {an essential} spectral gap. In particular, $\mathrm{Sf}^+$  {denotes} the spectral flow across any fiducial curve in the upper gap.
{Generally,} the spectral flow taken along any loop in $\mathring{C}$ could serve as a possible {edge} topological invariant. We actually show that all possible spectral flows are already determined by the one along $\mathcal{C}_R$. {More precisely,  $\pi_1(\mathring{C})$ is generated by two loops $\ell_+$ and $\ell_-$ oriented along increasing $a$ (Figure \ref{fig:cylinder}), and the main result is}
{
\begin{prop}\label{cor:bc.pumping}
	In the shallow-water wave model, the spectral flow structure is 
	\begin{equation*}
		{\rm Sf}^+(\ell_+)=-{\rm Sf^+}(\ell_-)=-\frac{C_+}{2}=-1,
	\end{equation*}
	where the loop $\ell_\pm$ can be taken at any fixed $k_x\gtrless 0$.
\end{prop}}
Thus our bulk-edge correspondence, Theorem \ref{thm:naive_BEC}, or the perturbed version Corollary \ref{cor:perturbed}, gives a \emph{complete} correspondence of bulk and {edge} topological invariants for the shallow-water model, {with boundary conditions given in Eq.~\eqref{eq:boundary_condition}}. 

{
\begin{rem}\label{rk:fixed_kx}
When {$k_x>0$ is fixed and} the boundary condition $a$ is varied once around a complete loop, there is pumping of a positive and fixed momentum edge mode from the upper band to the flat band; the pumping goes in the reverse direction for negative momentum modes. {This is another instance of charge pumping via a cycle of boundary conditions, as introduced in Section \ref{sec:Laplace.pump}.}	Loops at fixed $k_x$ such as $\ell_\pm$ are probably more easy to implement than $\ell_0$ in practice.
\end{rem}
}

\section{Continuity of the self-adjoint family {in the half-space model} \label{sec:proof_continuity} }

In this Section we prove Theorem \ref{thm:gap.continuity} {in two steps}. We first provide a continuous parametrization of all the abstract self-adjoint extensions of $H^\sharp(k_x)$ via the theory of von Neumann. Then we {derive Eq.\ \eqref{eqn:embedding.formula},}  {showing how} the parametrization by boundary conditions \eqref{eq:boundary_condition} is continuously embedded into the abstract scheme.

\subsection{Universal family of self-adjoint extensions}\label{sec:sa.ext}
We first decompose  $H^\sharp(k_x) = H^\sharp_0 + B(k_x)$ with
\begin{equation}\label{eqn:bounded.pert}
H^\sharp_0 := \begin{pmatrix} 0 & 0 & -\ii\frac{d}{dy} \\ 0 & 0 & -\ii\nu\frac{d^2}{dy^2} \\ -\ii\frac{d}{dy} & \ii\nu\frac{d^2}{dy^2} & 0 \end{pmatrix}, \qquad B(k_x) := \begin{pmatrix} 0 & k_x & 0 \\ k_x & 0 & -\ii f+\ii\nu k_x^2\\ 0 & \ii f-\ii\nu k_x^2 & 0\end{pmatrix}.
\end{equation}
Since $B(k_x)$ is bounded, it suffices to find the domains of self-adjointness for $H^\sharp_0$ to obtain those for $H^\sharp(k_x)$. 
The von Neumann theory (see X.1 of \cite{RS2} for a pedagogical treatment) provides a systematic way to do so.

As a preliminary simplification, observe that with the scaling substitution $\tilde{y}=\frac{y}{2\nu}$, we have {the formal differential operator}
\begin{equation*}
	H^\sharp_0=\frac{1}{2\nu}\begin{pmatrix} 0 & 0 & -\ii\frac{d}{d\tilde{y}} \\ 0 & 0 & -\frac{\ii}{2}\frac{d^2}{d\tilde{y}^2} \\ -\ii\frac{d}{d\tilde{y}} & \frac{\ii}{2}\frac{d^2}{d\tilde{y}^2} & 0 \end{pmatrix}.
\end{equation*}
{This defines a symmetric operator on the initial domain $C_c^\infty(\RR_+;\CC^3)\subset L^2(\RR_+;\CC^3)$, where $\RR_+=(0,\infty)$ denotes the half-line, and we take $H^\sharp_0$ to be the closure.} The adjoint $(H^\sharp_0)^*$ is an extension of $H^\sharp_0$ to the domain $H^1(\mathbb R_+)\oplus H^2(\mathbb R_+)\oplus H^2(\mathbb R_+)$ (the Sobolev spaces on $\mathbb R_+$) with no boundary condition imposed whatsoever. This {latter} domain is too large in the sense that $(H^\sharp_0)^*$ has $L^2$-eigenvalues off the real axis. That is, there are non-trivial \emph{deficiency subspaces}
\begin{align*}
	\mathcal{V}_{\pm}:={\rm ker}\left((H^\sharp_0)^*\mp \frac{\ii}{2\nu}\right)={\rm Ran}\left(H^\sharp_0\pm \frac{\ii}{2\nu}\right)^\perp,
\end{align*}
as captured by the \emph{deficiency indices}\footnote{The scaling by $2\nu$ is convenient for our problem and does not change the deficiency indices (\cite{RS2}, Theorem X.1).} $n_\pm={\rm dim}\,\mathcal{V}_\pm$, which we now compute.

The $+\ii/2\nu$ eigenfunctions of $(H^\sharp_0)^*$ can be found by the exponential ansatz $e^{-\lambda y}$, with the characteristic roots $\lambda$ satisfying
\begin{equation}
	0=4\nu^4\lambda^4-4\nu^2\lambda^2+1=(2\nu^2\lambda^2-1)^2.\label{eqn:char.roots}
\end{equation}
For normalizability, we need the positive root $1/\sqrt{2}\nu$, which is repeated, and conclude that $n_+=2$. That $n_-=2$ is deduced in a similar way.
\begin{prop}\label{prop:self.adjoint.domain.independence}
	The deficiency indices of $(H^\sharp_0)^*$ are $(n_+,n_-)=(2,2)$, so that the domains of the self-adjoint extensions of $H^\sharp_0$, thus also those of $H^\sharp(k_x)$, are in bijection with ${\rm U}(2)$.
\end{prop}
In more detail, given a unitary isomorphism $U:\mathcal{V}_+\rightarrow \mathcal{V}_-$, we define the operator $H^\sharp_0(U)$ acting on the {extended} domain
\begin{align}
	{\rm Dom}(H^\sharp_0(U))&=\{\phi+\psi_++U\psi_+\;|\;\phi\in{\rm Dom}(H^\sharp_0), \psi_+\in \mathcal{V}_{+}\},\nonumber\\
	H^\sharp_0(U)(\phi + \psi_+ + U\psi_+)&= H^\sharp_0(\phi) + \frac{\ii}{2\nu}\psi_+ -\frac{\ii}{2\nu} U\psi_+.\label{eqn:extension.formula}
\end{align}
Then $H^\sharp_0(U)$ is a self-adjoint extension of $H^\sharp_0$, and all self-adjoint extensions of $H^\sharp_0$ are of this form for some unitary $U$. We mention that a choice of orthonormal basis for $\mathcal{V}_+$ and for $\mathcal{V}_-$ is needed in order to identify $U$ with a matrix in ${\rm U}(2)$. 

Since $H^\sharp_0(U)$ is now self-adjoint, the resolvent $(H^\sharp_0(U)\pm \ii)^{-1}$ exists as a bounded operator from $L^2((0,\infty);\CC^3)\rightarrow {\rm Dom}(H^\sharp_0(U))$.
\begin{lem}\label{lem:sa.para.continuous}
	The map ${\rm U}(2)\ni U\mapsto (H^\sharp_0(U)\pm \ii)^{-1}$ is norm continuous. 
\end{lem}
\begin{proof}
An equivalent statement is the continuity of $U\mapsto (H^\sharp_0(U)+ \ii/2\nu)^{-1}$ (\cite{RS1} Theorem VIII.19). There is an orthogonal splitting of the Hilbert space as ${\rm Ran}(H^\sharp_0+ \ii/2\nu)\oplus \mathcal{V}_{+}$. The domain of $H^\sharp_0(U)$ is the vector space sum ${\rm Dom}(H^\sharp_0)+I(U)$, where $
	I(U):=\{\psi_++U\psi_+\,|\,\psi_+\in\mathcal{V}_+\}.
	$
	From its construction, Eq.\ \eqref{eqn:extension.formula}, we have
	\begin{equation*}
		(H^\sharp_0(U)+\ii/2\nu)(\phi+\psi_++U\psi_+)=(H^\sharp_0+\ii/2\nu)\phi+(\ii/\nu)\psi_+.
	\end{equation*}
	Thus $(H^\sharp_0(U)+\ii/2\nu)^{-1}$ maps the first component ${\rm Ran}(H^\sharp_0+ \ii/2\nu)$ back to ${\rm Dom}(H^\sharp_0)$, independently of the choice of $U$. On the second component $\mathcal{V}_+$, the resolvent takes $\psi_+$ to $\frac{\nu}{\ii}(1+U)\psi_+\in I(U)$, whence we see that its dependence on $U$ is (norm-)continuous. 
\end{proof}

\begin{defn}\label{defn:sa.extension}
	For $U\in{\rm U}(2)$, we write $H^\sharp_U(k_x)$ for the self-adjoint extension of $H^\sharp(k_x)$ given by $H^\sharp_U(k_x):=H^\sharp_0(U)+B(k_x)$ (the right side is defined in Eq.\ \eqref{eqn:extension.formula} and \eqref{eqn:bounded.pert}), acting on the domain
	\begin{equation}
		{\rm Dom}(H^\sharp_U(k_x))={\rm Dom}(H^\sharp_0(U))=\{\phi+\psi_++U\psi_+\;|\;\phi\in{\rm Dom}(H^\sharp_0), \psi_+\in \mathcal{V}_{+}\}.\label{eqn:sa.domain}
	\end{equation}
\end{defn}

\begin{prop}\label{prop:resolvent.continuity}
	The resolvent $(H^\sharp_U(k_x)\pm \ii)^{-1}$ depends (jointly) continuously on $U\in{\rm U(2)}$ and on $k_x\in\RR$.  
\end{prop}
\begin{proof}
	Suppose $(U^\prime,k_x^\prime)\rightarrow (U,k_x)$. Using a standard identity for the resolvent of a perturbed operator, we get the desired convergence,
	\begin{align*}
		\left(H^\sharp_{U^\prime}(k_x^\prime)+\ii \right)^{-1}&\equiv  \left(H^\sharp_0(U^\prime)+B(k_x^\prime) +\ii \right)^{-1}\\ 
		&=\left(H^\sharp_0(U^\prime)+\ii\right)^{-1}\left(1+B(k_x^\prime)(H^\sharp_0(U^\prime)+\ii)^{-1}\right)^{-1}\\
		&\longrightarrow  \left(H^\sharp_0(U)+\ii\right)^{-1}\left(1+B(k_x)(H^\sharp_0(U)+\ii)^{-1}\right)^{-1}\equiv\left(H^\sharp_{U}(k_x)+\ii \right)^{-1}.
	\end{align*}
	Here we used continuity of $k_x\mapsto B(k_x)$, and $U\mapsto (H^\sharp_0(U)+\ii)^{-1}$ (Lemma \ref{lem:sa.para.continuous}), as well as joint continuity of algebraic operations on bounded operators. Similarly, $(H^\sharp_{U^\prime}(k_x^\prime)-\ii)^{-1}\rightarrow (H^\sharp_{U}(k_x)-\ii)^{-1}$.
\end{proof}

\subsection{Sub-family {parametrized by} self-adjoint boundary conditions}
We proceed to identify which of the abstract self-adjoint domains in Section \ref{sec:sa.ext} (labelled by $U$) are realized by the concrete boundary conditions of \eqref{eq:boundary_condition} (labelled by $(k_x,a)$). As preparation, we work out the deficiency subspaces $\mathcal{V}_\pm$ explicitly. 

We may verify that
\begin{equation*}
	\psi_{1,\pm}(\tilde{y})=\begin{pmatrix} \sqrt{2}\tilde{y}-1 \\ \sqrt{2}-\tilde{y} \\ \pm \tilde{y} \end{pmatrix} e^{-\sqrt{2}\tilde{y}},\qquad \psi_{2,\pm}(\tilde{y})=\begin{pmatrix}\sqrt{2}\tilde{y} + 1 \\ -\tilde{y} \\ \pm (\tilde{y} + \sqrt{2})\end{pmatrix}e^{-\sqrt{2}\tilde{y}},
\end{equation*}
are $\pm \ii/2\nu$ eigenfunctions of $(H^\sharp_0)^*$, so that $\{\psi_{1,\pm},\psi_{2,\pm}\}$ spans $\mathcal{V}_\pm$. This pair of bases may not look optimal, but we will shortly see why it was chosen in this way.

Notice that $\langle \psi_{i,+}|\psi_{j,+}\rangle=\langle\psi_{i,-}|\psi_{j,-}\rangle$ for all $i,j=1,2$. So for each $\beta\in{\rm U}(1)$, the transformation of basis vectors,
\begin{equation*}
	U_\beta\psi_{1,+}:=\beta \psi_{1,-},\qquad U_\beta\psi_{2,+}:=\psi_{2,-},
\end{equation*}
defines a unitary map $U_\beta:\mathcal{V}_+\rightarrow\mathcal{V}_-$ . In turn, we obtain (through Definition \ref{defn:sa.extension}) a ${\rm U}(1)$-subfamily of self-adjoint extensions, $\beta\mapsto H^\sharp_{U_\beta}(k_x)$. This subfamily may be understood through boundary conditions as follows.

The basis vectors $\psi_{2,\pm}\equiv(\eta_{2,\pm},u_{2,\pm},v_{2,\pm})$ were engineered to ensure that $\psi_{2,+}+U_\beta\psi_{2,+}=\psi_{2,+}+\psi_{2,-}$ has {the following} boundary conditions in the last two {$u,v$} components,
\begin{equation*}
	(u_{2,+}+u_{2,-})|_{\tilde{y}=0}=(v_{2,+}+v_{2,-})|_{\tilde{y}=0}=(v^\prime_{2,+}+v^\prime_{2,-})|_{\tilde{y}=0}=0,
\end{equation*}
{while its} first component $(\eta_{2,+}+\eta_{2,-})|_{\tilde{y}=0}$ {is unrestricted}. This means that {for functions} in the vector space 
\begin{align*}
	I(U_\beta)&=\{\psi_++U_\beta\psi_+\;:\;\psi_+\in\mathcal{V}_+\}\\
	&=\{c_1(\psi_{1,+}+\beta\psi_{1,-})+c_2(\psi_{2,+}+\psi_{2,-})\,:\,c_1,c_2\in\CC\},
\end{align*}
{only the $c_1(\psi_{1,+}+\beta\psi_{1,-})$ part contributes to the boundary conditions,}
\begin{equation}
	(\psi_++U_\beta\psi_+)|_{\tilde{y}=0}=c_1(\psi_{1,+}+\beta\psi_{1,-})|_{\tilde{y}=0}=c_1\begin{pmatrix}(\sqrt{2}\tilde{y}-1)(1+\beta) \\ (\sqrt{2}-\tilde{y})(1+\beta) \\ \tilde{y}(1-\beta)\end{pmatrix}e^{-\sqrt{2}\tilde{y}}\Big|_{\tilde{y}=0}.\label{eq:second.contribution}
\end{equation}
Boundary conditions for the functions {$(\eta,u,v)^{\rm t}$ belonging to}
\begin{equation*}
	{\rm Dom}(H^\sharp_{U_\beta}(k_x))={\rm Dom}(H^\sharp_0({U_\beta}))={\rm Dom}(H^\sharp_0)+I(U_\beta)
\end{equation*}
come from the $I(U_\beta)$ part, and these {conditions} are read off from Eq.\ \eqref{eq:second.contribution} as
\begin{equation}
	v|_{y=0}=0,\qquad (1-\beta)u|_{y=0}=\sqrt{2}(1+\beta)v^\prime|_{y=0},\qquad\eta|_{y=0}\;{\rm unrestricted}.
	\label{eq:abstract.boundary.condition}
\end{equation}

Therefore, any boundary condition expressed in terms of $(k_x,a)$ in Eq.\ \eqref{eq:boundary_condition} has an equivalent expression in terms of some $\beta=\beta(k_x,a)$ in Eq.\ \eqref{eq:abstract.boundary.condition}, via the relation
\begin{equation}
	\frac{\ii a}{\sqrt{2}k_x}=\frac{u|_{y=0}}{\sqrt{2}v^\prime|_{y=0}}=:\frac{1+\beta(k_x,a)}{1-\beta(k_x,a)}\;\;\Leftrightarrow\;\; H^\sharp(k_x,a)=H^\sharp_{U_{\beta(k_x,a)}}(k_x).
	\label{eq:bc.parameter.change}
\end{equation}

Recall the Cayley transform homeomorphisms
\begin{align*}
{\rm Cay}&:\mathbb{R}\cup\{\infty\}\rightarrow {\rm U}(1),\qquad \gamma\mapsto \frac{\gamma-\ii}{\gamma+\ii},\\
{\rm Cay}^{-1}&:{\rm U}(1)\rightarrow\mathbb{R}\cup\{\infty\},\qquad \beta\mapsto \ii\frac{1+\beta}{1-\beta}.
\end{align*}
The relation in Eq.\ \eqref{eq:bc.parameter.change} is uniquely satisfied by
\begin{equation*}
	\beta(k_x,a)={\rm Cay}\left(-\frac{a}{\sqrt{2}k_x}\right)=\frac{a+\ii\sqrt{2}k_x}{a-\ii\sqrt{2}k_x}.
\end{equation*}
Note that the above formula also works in the $k_x=0$ case, where every $a\neq 0$ defines the same boundary condition corresponding to $\beta(0,a)=1$. (Recall that $(k_x,a)=(0,0)$ is inadmissible as a boundary condition.) Altogether, we have constructed the continuous classifying map
\begin{align}
	U_\beta:\mathring{C}&\rightarrow {\rm U}(1)\subset {\rm U}(2)\nonumber\\
	(k_x,a)&\overset{\beta}{\mapsto} \frac{a+\ii\sqrt{2}k_x}{a-\ii\sqrt{2}k_x}\mapsto U_{\beta(k_x,a)}.\label{eq:bc.to.extension}
\end{align}

\medskip

\begin{proof}[Proof of Theorem \ref{thm:gap.continuity}]
By definition, the assignment $(k_x,a)\mapsto H^\sharp(k_x,a)$ factors as
\begin{equation}
(k_x,a)\mapsto U_{\beta(k_x,a)}\mapsto H^\sharp_{U_{\beta(k_x,a)}}(k_x)\equiv H^\sharp(k_x,a),\label{eqn:embedding.formula}
\end{equation}
so the continuity result follows immediately from Prop.\ \ref{prop:resolvent.continuity}.
\end{proof}

\begin{rem}\label{rem:bc.parameter.winding}
The map $\beta$ in Eq.\ \eqref{eq:bc.to.extension} restricts to a winding number $\mp 1$ map when $k_x\gtrless 0$, and degenerates to the constant map 1 at $k_x=0$. The puncture at $(k_x,a)=(0,0)$ means that at $k_x=0$, the map $\beta$ is only defined along a punctured circle $(\RR\cup\{\infty\})\setminus\{0\}$, and has ill-defined winding number. This allows the winding number of $\beta(k_x,\cdot)$ to ``continuously'' switch signs when $k_x$ switches sign. 
\end{rem}

\section{Bulk-edge correspondence via scattering theory \label{sec:BEC}}
In this section we prove Theorem~\ref{thm:naive_BEC} using scattering theory. Such an approach was developed in \cite{GJT21} for a fixed value of $a$, so we extend it here to an $a$-dependent framework. The strategy is the following: We first define the scattering amplitude $S$ and establish its main properties. Then we relate the number of edge modes with the winding number of $S$ along some limiting loop approaching $\mathcal C_R$. Furthermore, the Chern number $C_+$ is also related to a winding number of $S$ along some other loop, which is independent of $a$. Finally we show that the two winding numbers coincide. 

Before that, we reformulate the number of edge modes in a more suitable way. 
\begin{defn}\label{defn:edge.index}
The number of edge modes $n_\mathrm{b}(R)$ below a bulk spectral band is the signed number of
	edge mode branches emerging (counted positively) or disappearing (counted negatively) at the lower band limit, as $\theta$ increases. The number $n_\mathrm{a}(R)$ of edge modes above a band is counted likewise up to a global	sign change.
\end{defn}

The family $H_{R,\theta}^\sharp$ has three bulk bands, see \eqref{eq:esspec_Hsharp_theta}, that we denote by $+,\,-$ and $0$. The definition above is more general than Definition~\ref{defn:edge.index.naive} and works beyond compact parameters, see \cite{GJT21}. However, the parameter $\theta$ is compact here and since the edge mode branches are continuous, the two definitions coincide in our case. One has 
\begin{equation}
n_\mathrm{b}^+(R) = n(\mathcal C_R,\mu) = n_\mathrm{a}^0(R).\label{eq:fiducial.versus.straight}
\end{equation}
One could think of $n_\mathrm{b}^+(R)$ as a crossing counting with a fiducial line $\omega=\mu$ that has been continuously deformed to the bottom of the $+$ band curve $\mu_\theta = \omega_\theta$ (see Figure \ref{fig:spectral_flow}). The number $n_b$ can then be related to the Chern number via scattering theory. Moreover, notice that the middle band $\omega_0=0$ is completely flat so that scattering theory cannot be developed there. However, because it is also topologically trivial ($C_0=0$), we can simply ignore it and reduce the proof of Theorem~\ref{thm:naive_BEC} to the relation $n_\mathrm{b}^+(R) = C_+$. Finally, notice that Section~\ref{sec:sf.structure}  discusses the possibility of extending Definition~\ref{defn:edge.index} via Phillips spectral flow across an energy curve $\mu_\theta$.

\subsection{Scattering amplitude} 

This subsection mostly imports from \cite{GJT21} the required formalism such as bulk sections, scattering state and amplitude. We include it for self-consistency of the paper, to set notations and to provide explicit expressions that are used later. We also emphasize the $a$ dependence. 

\paragraph{Bulk data.} A Fourier decomposition of bulk equation \eqref{eq:ShallowWater_Schrodinger} with modes $\psi =\widehat\psi \ee^{\ii(k_x x + k_y y -\omega t)}$ leads to the eigenvalue equation:
\begin{equation}\label{eq:ShallowWater_bulk}
	H \widehat \psi = \omega \widehat \psi, \qquad \widehat \psi = \begin{pmatrix}
		\hat \eta \\ \hat u\\ \hat v
	\end{pmatrix} \qquad  H({k_x,k_y}) = \begin{pmatrix}
		0 & k_x & k_y \\ k_x & 0 & -\ii (f-\nu  {k}^2) \\ k_y & \ii (f-\nu  {k}^2) & 0
	\end{pmatrix},
\end{equation}
with ${k}^2=k_x^2+k_y^2$ and $H({k_x,k_y})$ a Hermitian matrix. It admits three frequency bands $\omega_\pm$ and $\omega_0$ given in \eqref{eq:bulkbands}. We shall focus on the upper band $\omega_+$ from now on. With $\nu>0$, the problem can be compactified at $k\to \infty$, and we identify the compactified $k$-plane with the Riemann sphere $\mathbb C \cup \{\infty\} \cong S^2$ via $z =k_x+\ii k_y\equiv (k_x,k_y)$. We call $\wpsi(k_x,k_y)$ a bulk section. Since $C_+=2$, it is impossible to find a global section that is regular for all $z \in S^2$. We need at least two distinct ones, that are regular locally on two overlapping patches to cover the sphere. One section is given by
\begin{equation}\label{eq:section_infty}
	\widehat \psi^\infty(k_x,k_y) = \dfrac{1}{k_x- \ii k_y} \begin{pmatrix}
		{k}^2/\omega_+ \\ k_x-\ii k_y q \\ k_y+ \ii k_x q 
	\end{pmatrix}\,, \quad 
	q(k_x,k_y) := \tfrac{f-\nu {k}^2}{\omega_+}\,,\quad \omega_+=\omega_+(k_x,k_y)\,.
\end{equation}
Notice that $q \rightarrow 1$ (resp. $-1$) as $k\rightarrow 0$ (resp. $\infty$). Thus \eqref{eq:section_infty} defines a section of the eigenbundle of $\omega_+$ that is smooth and non zero for all $z \in \mathbb C$, including $z=0$, but not at $\infty$, where it is singular and winds like $z/\bar{z}$. We can move the singularity by defining for $\zeta=\zeta_x + \ii \zeta_y \in \mathbb C$,
\begin{equation}\label{eq:section_zeta}
	\widehat \psi^\zeta = t^\zeta_\infty \widehat \psi^\infty, \qquad t^\zeta_\infty(z)=\frac{\bar{z}-\bar{\zeta}}{z-\zeta}.
\end{equation}
Notice that $\widehat \psi^\zeta$ is regular on $S^2\setminus\{\zeta\}$. On the other hand, $t^\zeta_\infty$ is regular for all $z \in S^2\setminus\{\infty,\zeta\}$ and singular at the two omitted points.


\paragraph{Scattering amplitude.} Back in the half-space, translation invariance is broken along $y$, so normal modes  $\psi =\widehat\psi \ee^{\ii(k_x x + k_y y -\omega t)}$ are not solutions of the boundary problem, and $k_y$ is not conserved. However for $\kappa>0$ we shall consider such modes with $k_y=\kappa$ and $k_y=-\kappa$ as outgoing and incoming plane waves with respect to the boundary. Such modes have the same frequency since $\omega(k_x,\kappa) = \omega(k_x,-\kappa)$. There are actually two other values of $k_y$ with the same frequency, however they are purely imaginary and read
\begin{equation}\label{eq:def_kappa_ev}
	\kappa_\mathrm{\mathrm{ev}/\mathrm{div}}(k_x,\kappa) = \pm \ii \sqrt{\kappa^2 + 2 k_x^2 + \dfrac{1-2\nu f}{\nu^2}} \in \pm \ii \mathbb R_+\,,
\end{equation}
When $k_y=\kappa_\mathrm{ev}$ the evanescent mode is exponentially decaying as $y\to \infty$, whereas the diverging mode for $k_y=\kappa_\mathrm{div}$ diverges. 

\begin{rem}
	Another perspective is to look at the solutions  of
	$$
	 k_x^2 + k_y^2 + (f-\nu(k_x^2+k_y^2))^2 = \omega^2
	$$ for fixed $k_x \in \mathbb R$ and $\omega^2 > k_x^2 + (f-\nu k_x^2)^2$. Such an equation has four solutions: two are real and two purely imaginary. It turns out that they can all be expressed in terms of the real positive one. If we call the latter $\kappa >0$ we recover the expressions above.
\end{rem}

Let $U_{\mathrm{out}}\subset\mathbb{R}^2$ be an open subset, and let $U_{\mathrm{in}}\subset\mathbb{R}^2$ and $U_{\mathrm{ev}}\subset\mathbb{R}\times\ii \mathbb{R}$ be the images under the maps $(k_x,\kappa)\mapsto(k_x,-\kappa)$ and $(k_x,\kappa)\mapsto (k_x,\kappa_{\mathrm{ev}})$. Consider bulk section $\wpsi_\mathrm{in/out/ev}$ with momentum $k_x$ and $k_y= -\kappa$, $\kappa$ and $\kappa_\mathrm{ev}$ for $k_x,\kappa \in U_{\mathrm{out}}$. We assume the sections to be non-vanishing and regular on their respective domains, and that they are of amplitude one, namely $\langle \psi_{\mathrm{in}}(k_x,-\kappa),  \psi_{\mathrm{in}}(k_x,-\kappa) \rangle =1$ and similarly for out and ev.

The scattering state is the linear combination 
\begin{equation}\label{eq:scattering_solution}
	\psi_s = \big(\wpsi_\mathrm{in}\ee^{-\ii \kappa y} + S\wpsi_\mathrm{\mathrm{out}}\ee^{\ii \kappa y} + T\wpsi_\mathrm{ev} \ee^{\ii \kappa_\mathrm{ev} y}\big)\ee^{\ii (k_x x-\omega t)}
\end{equation}
which satisfies the boundary condition \eqref{eq:boundary_condition}. It has well defined momentum $k_x$ and frequency $\omega = \omega_+(k_x,\kappa)$. Such a state is uniquely defined up to multiples, and  for any self-adjoint boundary condition, the quantity $S(k_x,\kappa) \in U(1)$ is called the scattering amplitude, see \cite{GJT21}. It depends on the choice of bulk sections and on the boundary condition, so we shall rather write $S(k_x,\kappa,a)$.

\paragraph{Explicit expression and properties.}
For $\kappa >0$ we shall work with $\zeta = \ii$ (any  $\zeta = \ii \zeta_y$ with $\zeta_y>0$ would fit)  and 
\begin{equation}\label{S_explicit}
\wpsi_\mathrm{in} (k_x,-\kappa) = \wpsi^\zeta(k_x,-\kappa), \qquad \wpsi_\mathrm{out} (k_x,\kappa) = \wpsi^\zeta(k_x,\kappa),\qquad \wpsi_\mathrm{ev}(k_x,\kappa_{\mathrm{ev}}) = \wpsi^\infty(k_x,\kappa_{\mathrm{ev}})
\end{equation}
so that $U_{\mathrm{out}} = S^2\setminus\{\zeta\}$, $U_{\mathrm{in}} = S^2\setminus\{-\zeta\}$ and $U_\mathrm{ev} = S^2$. In particular $U_{\mathrm{in}} \cup U_{\mathrm{out}} = S^2$ and $U_{\mathrm{in}} \cap U_{\mathrm{out}} = S^2\setminus\{\pm\zeta\}$. Also notice that $\wpsi_{\mathrm{ev}}$ is regular for all $(k_x,\kappa) \in S^2$ (yet it is not a global section on $S^2$ because it has momentum $k_x,\kappa_{\mathrm{ev}} \neq k_x,\kappa$). Unless stated, we shall always work with this choice of sections to define $S$. 

For $\zeta \in S^2$ we write
$$
\wpsi^\zeta = \begin{pmatrix}
	\eta^\zeta \\
	u^\zeta\\
	v^\zeta
\end{pmatrix}, 
$$
so that the boundary condition \eqref{eq:boundary_condition} implies for the scattering state
\begin{align}
	&v^\zeta(-\kappa)+ S v^\zeta(\kappa)  + T v^\infty(\kappa_\mathrm{ev})  =0\\
	&k_x u^\zeta(-\kappa)-a \kappa v^\zeta(-\kappa) + S (k_x u^\zeta(\kappa)+a \kappa v^\zeta(\kappa))  + T (k_x u^\infty(\kappa_\mathrm{ev})+a \kappa_\mathrm{ev} v^\infty(\kappa_\mathrm{ev}) ) =0
\end{align}
where we omitted the $k_x$ dependence,
leading to
$$
	S(k_x,\kappa,a) = - \dfrac{g(k_x,-\kappa,a)}{g(k_x,\kappa,a)}, $$
where
\begin{equation}	\label{defg}
g(k_x,\kappa,a)=	\begin{vmatrix}
k_x u^\zeta(k_x,\kappa)+a \kappa v^\zeta(k_x,\kappa) & k_x u^\infty(k_x,\kappa_\mathrm{ev})+a \kappa_\mathrm{ev} v^\infty(k_x,\kappa_\mathrm{ev})\\
v^\zeta(k_x,\kappa) & v^\infty(k_x,\kappa_{\mathrm{ev}})
	\end{vmatrix}.
\end{equation}
Recall that $\kappa_{\mathrm{ev}}$ depends on $k_x,\kappa$, see \eqref{eq:def_kappa_ev}. It was shown in \cite{GJT21} that $S$ is well-defined (in particular $g$ is non-vanishing) as long as the boundary condition is self-adjoint and the sections $\wpsi_{\mathrm{in}}$ and $\wpsi_{\mathrm{out}}$ are regular. Thus $S(k_x,\kappa,a)$ is defined on 
$$
D_S = \Big(\mathbb R \times \mathbb  R_+^* \times \mathbb R \Big) \setminus \Big( \{(0,\kappa,0) \, | \, \kappa >0\} \cup \{ (0,\zeta_y,a) \, | \, a \in \mathbb R\}\Big)
$$ 
illustrated in Figure~\ref{fig:winding_numbers} below. Moreover it is easy to see from its expression that $S$ is smooth on the domain $D_S$.

\subsection{Number of edge modes and relative Levinson's theorem}
 
In order to study a relevant scattering amplitude, consider the following curve for any $\epsilon >0$
$$
\mathcal C_R^\epsilon= \{ (k_x,\kappa,a) = (R \cos \theta,\epsilon, R \sin \theta) \quad | \quad \theta \in \mathbb [-\pi, \pi] \}.
$$

\begin{prop}\label{prop:relative.Levinson}
	For any $\epsilon$ such that $0<\epsilon < \zeta_y$, the scattering amplitude is well-defined and satisfies 
	$$
	\lim_{\epsilon \to 0} \dfrac{1}{2\pi \ii} \int_{\mathcal C_R^\epsilon} S^{-1} \dd S = n_\mathrm{b}^+(R).
	$$
\end{prop}

This is a consequence of the relative Levinson's theorem, see \cite[Thm 13]{GJT21}.  This general statement is true for any  {continuously parametrized compact domain}, such as {the loops $C_R^\epsilon$ parametrized by} $\theta$. As $\kappa \to 0$ the scattering amplitude is computed near the bottom of the upper band, and feels the bound states below it. As $\theta$ increases, it then counts the change in the number of bound states, that is the number of edge modes from the definition above. We illustrate it for specific values of $R$ in Figure~\ref{fig:relative_levinson}.

\begin{figure}[htb]
	\centering	
	\includegraphics[scale=0.45]{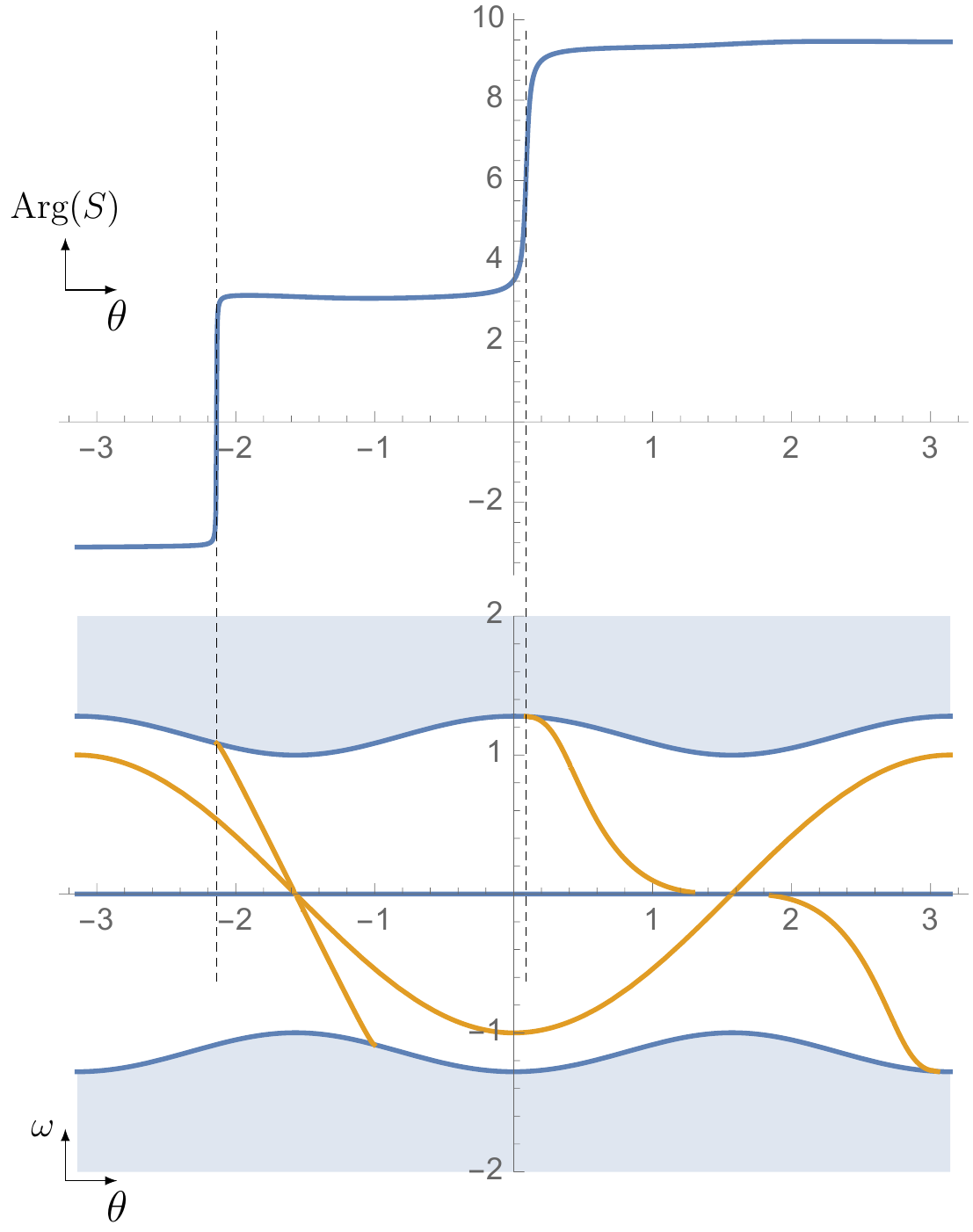} \includegraphics[scale=0.45]{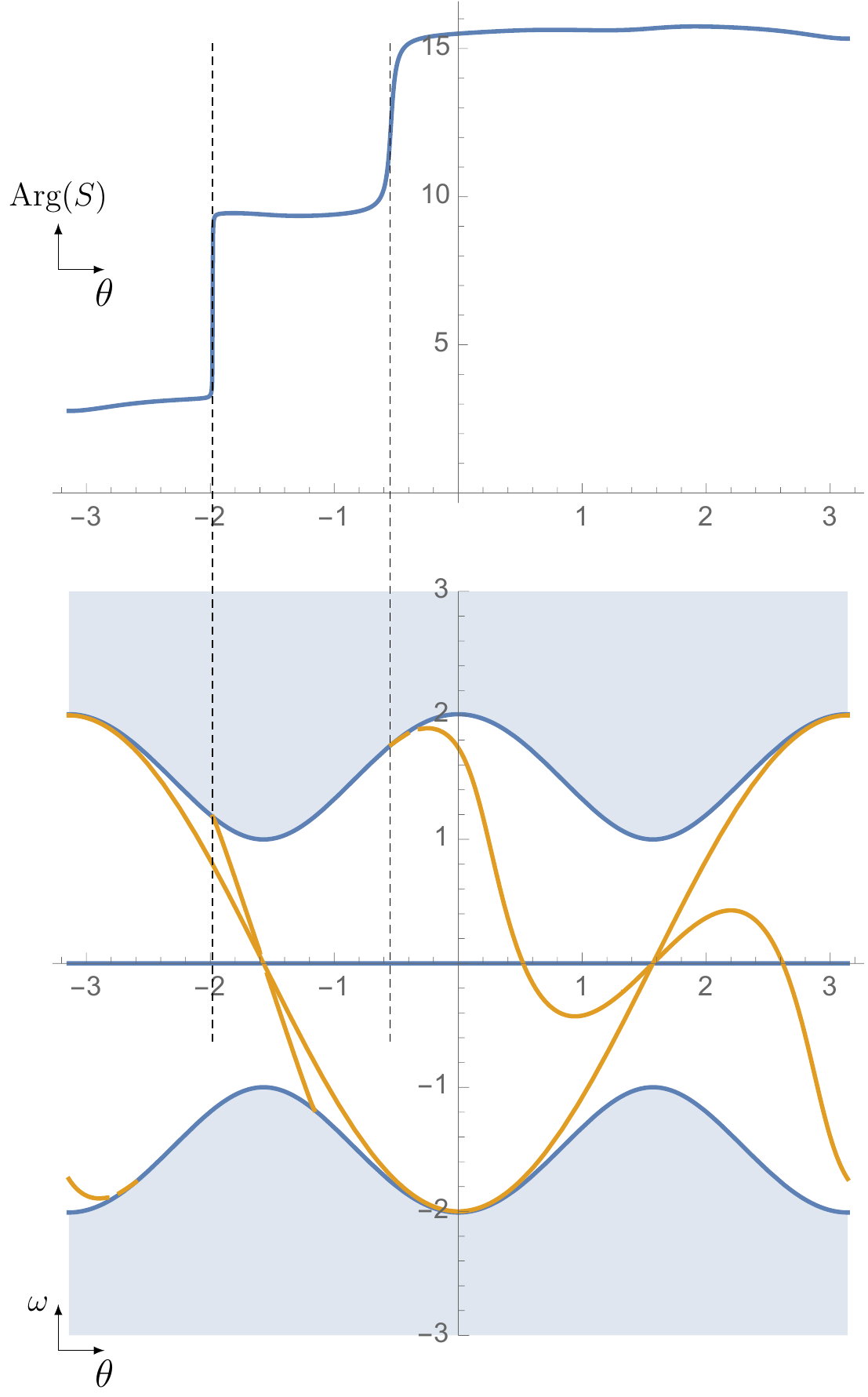} \includegraphics[scale=0.45]{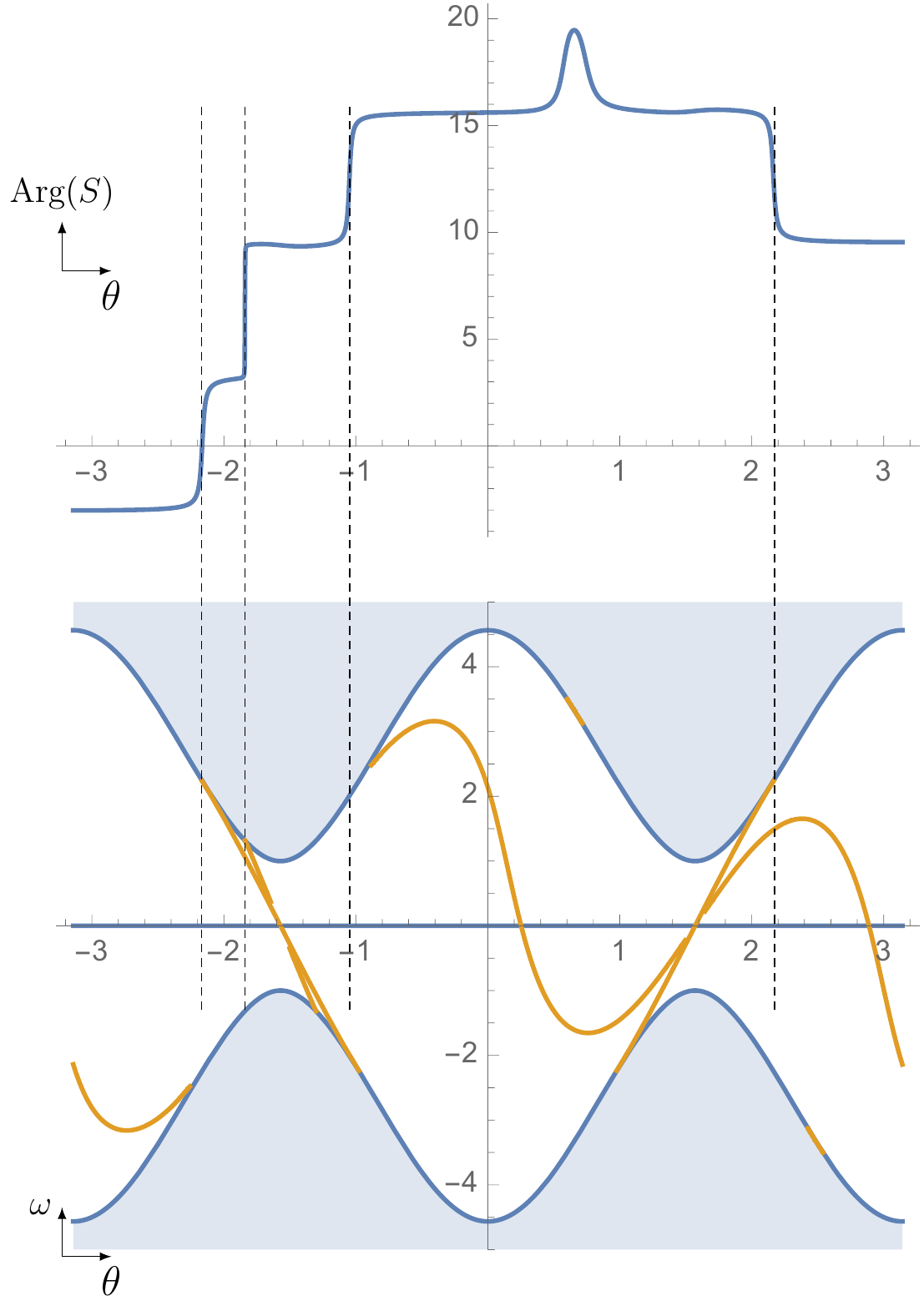}

	\caption{Lower panel : spectrum of $H^\sharp_{R,\theta}$ with respect to $\theta$ for $R=1$, $2$ and $4$ respectively.  Upper panel : argument of the scattering amplitude with respect to $\theta$ for $\epsilon = 0.05$. The argument jumps precisely when an edge mode branch disappears or emerges from the upper bulk band. The total change of argument of $S$ matches with the number of edge modes, that is always $2$.\label{fig:relative_levinson}}
\end{figure}

\subsection{Scattering and Chern number}

Here we consider a different curve in the parameter space. For $0<\delta<1$:
$$
\Gamma_{\delta,a_0} = \{(k_x,\kappa,a)  = (\delta\cos \alpha,\delta\sin \alpha + 1 ,a_0) \, | \, \alpha \in [-\pi,\pi]\}.
$$
This is a circle of radius $\delta$ in the $(k_x,\kappa)$ plane with $\kappa >0$, centered around $\zeta=(0,1)$. Any closed curve with that property would fit, in particular $\delta$ is not necessarily small here. Then we have
\begin{prop}
	For any $a_0 \neq 0$ and $\delta >0$,
	$$
	\dfrac{1}{2 \pi \ii } \int_{\Gamma_{\delta,a_0}} S^{-1}\dd S = C_+
	$$
	where $C_+$ is the Chern number of the upper band.
\end{prop}

This is called the bulk-scattering correspondence, and is proved in \cite[Thm 9]{GJT21}. The proof goes as follows. Pick a section $\wpsi_{\mathrm{in}}$ on $U_\mathrm{in} \cap \{\kappa <0\}$. The boundary condition required in the definition of the scattering state naturally defines a scattering map $\mathcal S : \mathcal E_{k_x,-\kappa} \to \mathcal E_{k_x,\kappa}$ where $\mathcal E_{k_x,\kappa}$ is the fiber above $(k_x,\kappa)$. Thus one naturally has the abstract section {denoted by $\widecheck \psi_{\mathrm{out}}$}, and given by $\widecheck \psi_{\mathrm{out}}(k_x,\kappa) = \mathcal S \wpsi_{\mathrm{in}}(k_x,-\kappa)$. {On the other hand, according to \eqref{S_explicit} and \eqref{eq:scattering_solution}} one uses a common reference section for in and out on $U_\mathrm{in} \cap U_\mathrm{out}$, namely $\wpsi_{\mathrm{in}} = \wpsi_{\mathrm{out}}$ so {that the abstract section is actually proportional to $\wpsi_{\mathrm{in}}$. Explicitly,} we have $\widecheck \psi_{\mathrm{out}}(k_x,\kappa) = S(k_x,\kappa,a) \wpsi_{\mathrm{in}}(k_x,-\kappa)$. Thus $S$ is nothing but a transition function between in and out sections. Consequently, if $U_\mathrm{in} $ and $U_\mathrm{out}$ cover $S^2$ and if the loop $\Gamma_{\delta,a_0}$ splits $S^2$ with interior contained in $U_\mathrm{in}$ and exterior contained in $U_\mathrm{out}$, then the winding of $S$ along it is the Chern number.  One can check that this is the case with the choice \eqref{S_explicit} for $S$ as long as $\zeta$ is inside $\Gamma_{\delta,a_0}$.

\subsection{The correspondence}\label{sec:spec.Chern}

The main claim in this section, together with the two propositions above, implies $n_b^+(R) = C_+$, and ends the proof of Theorem~\ref{thm:naive_BEC}. 

\begin{prop} The scattering amplitude $S$ satisfies:
$$
\dfrac{1}{2\pi \ii} \int_{\mathcal C_R^\epsilon} S^{-1} \dd S =	\dfrac{1}{2 \pi \ii } \int_{\Gamma_{\delta,a_0}} S^{-1}\dd S  
$$
for any $\epsilon$ sufficiently small for the left hand side to be well defined.
\end{prop} 

\begin{proof}
The scattering amplitude is defined explicitly through \eqref{S_explicit}, which ensures that it is well defined and smooth for $\kappa >0$ and away from $k_x=a=0$ and $k_x=0, \kappa = \zeta_y$. We summarize all that in the left panel of Figure~\ref{fig:winding_numbers}.

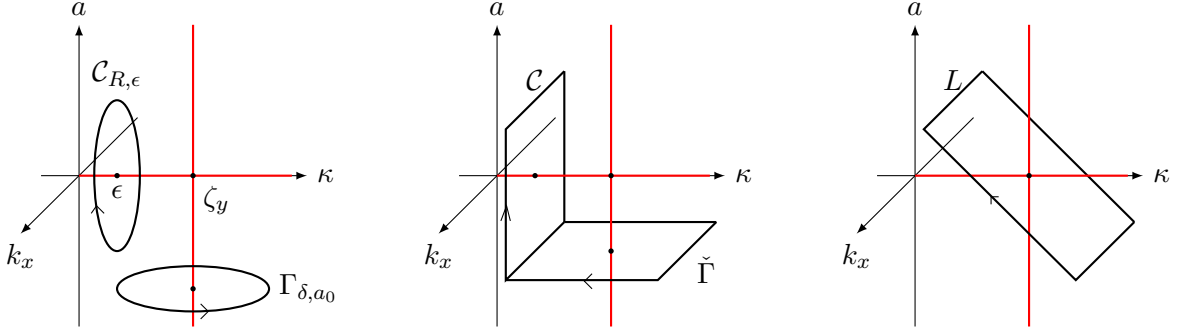
\begin{figure}[htb]
	\centering
\begin{tikzpicture}
	\draw[-latex] (-0.5,0,0) -- (3,0,0) node[right]{$\kappa$};
		\draw[-latex] (0,-2,0) -- (0,2,0) node[above]{$a$};
		\draw[-latex] (0,0,-2) -- (0, 0,2) node[below]{$k_x$};
		\draw[red,thick] (0,0,0) -- (2.8,0,0);
		\draw[red,thick] (1.5,-2,0) -- (1.5,2,0);
		\draw[thick] (0.5,0) circle[x radius=0.3cm, y radius=1cm];
		\fill (0.5,0) circle (1pt);
				\draw (0.5,0) node[below]{$\epsilon$};
		
		\draw[thick] (1.5,-1.5) circle[x radius=1cm, y radius=0.3cm];
		\fill (1.5,-1.5) circle (1pt);
		
		\draw (2.5,-1.5) node [right]{$\Gamma_{\delta,a_0}$};
		\draw (0.5,1) node [above]{$\mathcal C_{R,\epsilon}$};
		\draw (1.5,0) node[below right]{$\zeta_y$};
				\fill (1.5,0) circle (1pt);
				
		\draw (0.16,-0.5)--(0.23,-0.4) -- (0.33,-0.5);
		\draw (1.6,-1.7) -- (1.7,-1.8) -- (1.6,-1.9);
		\begin{scope}[xshift=5.5cm]
			
\draw[thick] (0.5,-1,-1) -- (0.5,1,-1);
\draw[thick] (2.5,-1,-1)-- (0.5,-1,-1);
			
	\draw[-latex] (-0.5,0,0) -- (3,0,0) node[right]{$\kappa$};
\draw[-latex] (0,-2,0) -- (0,2,0) node[above]{$a$};
\draw[-latex] (0,0,-2) -- (0, 0,2) node[below]{$k_x$};
\draw[red,thick] (0,0,0) -- (2.8,0,0);
\draw[red,thick] (1.5,-2,0) -- (1.5,2,0);

\fill (0.5,0) circle (1pt);

\fill (1.5,-1) circle (1pt);

\draw (2.5,-1.25) node [right]{$\check\Gamma$};
\draw (0.5,1) node [above]{$\mathcal C$};
\fill (1.5,0) circle (1pt);

\draw[thick] (0.5,-1,-1) -- (0.5,-1,1) -- (0.5,1,1) -- (0.5,1,-1);
\draw[thick] (0.5,-1,1) -- (2.5,-1,1) -- (2.5,-1,-1);

\draw (1.25,-1.3)  -- (1.15,-1.4) -- (1.25,-1.5) ;
\draw (0.05, -0.6) -- (0.125,-0.4) -- (0.175,-0.6);
		\end{scope}

			\begin{scope}[xshift=11cm]
				\draw[thick] (0.5,1,-1) -- (2.5,-1,-1) ;
				\draw[-latex] (-0.5,0,0) -- (3,0,0) node[right]{$\kappa$};
				\draw[-latex] (0,-2,0) -- (0,2,0) node[above]{$a$};
				\draw[-latex] (0,0,-2) -- (0, 0,2) node[below]{$k_x$};
				\draw[red,thick] (0,0,0) -- (2.8,0,0);
				\draw[red,thick] (1.5,-2,0) -- (1.5,2,0);

				\fill (1.5,0) circle (1pt);

				\draw[thick] (0.5,1,-1) -- (0.5,1,1) -- (2.5,-1,1) -- (2.5,-1,-1);
			\draw (0.5,1) node [above]{$L$};
			\draw   (1,-0.38) -- (1,-0.28) -- (1.1,-0.28);
				
			\end{scope}
\end{tikzpicture}
\caption{The winding of $S$ along $\mathcal C_R^\epsilon$ tends to $n_b^+(R)$, and the winding along $\Gamma_{\delta,a_0}$ gives the Chern number $C_+$. Up to some continuous deformation, the difference of such winding numbers is encoded in the one around the loop $L$, that turns out to be zero.\label{fig:winding_numbers}}
\end{figure}

Apart from the region where $S$ is singular (red lines and $\kappa \leq 0$), the curves can be continuously deformed leaving the winding numbers unchanged. Thus, without loss of generality, we can set $a_0=-1$, $\zeta_y=1$, $\epsilon <1$ and $\delta<1$, and then deform $\mathcal C_R^\epsilon$ to a square loop of center $(0,\epsilon,0)$ and length $1$ in the plane $\kappa = \epsilon$, called $\mathcal C$. Similarly, we deform $\Gamma_{\delta,a_0}$ to a square loop of center $(0,1,-1)$ and length $1$ in the plane $a=-1$, called $\Gamma$. $\check \Gamma$ is the same loop but with reversed orientation. See the middle panel of Figure~\ref{fig:winding_numbers}. Then we have
$$
\int_\mathcal{C} S^{-1} \dd S - \int_\Gamma S^{-1} \dd S  = \int_L S^{-1} \dd S
$$
where $L$ is a deformation of $\mathcal C \cup \check\Gamma$ which lies in the half-plane $\mathcal P$ defined by $\kappa+a=1$, $\kappa >0$ (any half-plane of the form $\lambda \kappa +a = \lambda$, $\kappa >0$ for some $\lambda >0$ would actually fit). See the right panel of Figure~\ref{fig:winding_numbers}. Thus the winding numbers of $S$ along $\mathcal C_R^\epsilon$ and $\Gamma_{\delta,a_0}$ coincide if the winding of $S$ along any loop in $\mathcal P$ around $(0,1,0)$ vanishes. Consider for example for $0<\alpha<1$
$$
\ell_\alpha = \big\{(\alpha \cos(\theta), 1-\alpha \sin(\theta), \alpha \sin(\theta) ) \, | \, \theta \in [0,2\pi] \big\},
$$
which is homotopic to $L$ in $\mathcal P$ up to orientation.

Now we expand \eqref{defg} on $\ell$  near $\alpha = 0$. We use the fact that, for $k_x = \alpha \cos(\theta)$ and $k_y = 1-\alpha \sin(\theta)$,
$$
\lim_{\alpha \to 0} u^\zeta =  \dfrac{f-\nu}{\sqrt{1+(f-\nu)^2}}\ee^{2 \ii \theta}, \qquad\lim_{\alpha \to 0} v^\zeta  = \ii \ee^{2 \ii \theta}.
$$
Moreover one has $
 \kappa_{\mathrm{ev}} \to  \ii \sqrt{1+\frac{1}{\nu^2}-\frac{2f}{\nu}} $ and 
$$
\lim_{\alpha\to 0} u^\infty(k_x,\kappa_{\mathrm{ev}}) = \dfrac{1-f\nu+\nu^2}{\nu\sqrt{1+(f-\nu)^2}}, \qquad \lim_{\alpha\to 0} v^\infty(k_x,\kappa_{\mathrm{ev}}) = \ii.
$$
Consequently, one has to the leading order in $\alpha$
$$
g \sim - \ii \alpha\ee^{2\ii \theta} (A \cos\theta - (B+\ii) \sin(\theta)) 
$$
where $A,B>0$ are positive constant depending on $f$ and $\nu$. We then perform a similar computation for $g$ but flip $\kappa \to -\kappa$ so that $k_x = \alpha \cos(\theta)$ and $k_y = -1+\alpha \sin(\theta)$. To leading order near $\alpha=0$ we get
$$
g \sim  \ii \alpha (A \cos\theta - (B-\ii) \sin(\theta)) 
$$
Recalling that $S(k_x,\kappa,a) = -\frac{g(k_x,-\kappa,a)}{g(k_x,\kappa,a)}$ we deduce
$$
\lim_{\alpha\to 0} S = \ee^{-2\ii \theta} \dfrac{A \cos\theta - (B-\ii) \sin(\theta)}{A \cos\theta - (B+\ii) \sin(\theta)}.
$$
Although this limit depends on $\theta$, it is easy to see that the winding number of the expression on the right hand side vanishes, and so does the winding of $S$ along any $\ell_\alpha$, which ends the proof. 

\end{proof}

\section{Full spectral flow structure \label{sec:sf.structure}}

As discussed in Section~\ref{sec:main_stability}, the stability of the edge index requires a generalization of spectral flow from Definition~\ref{defn:edge.index.naive} to take into account pathological edge mode branches that may occur when perturbing the system. 

Let $\mathscr{F}^{\rm sa}$ denote the space of self-adjoint (possibly unbounded) Fredholm operators, and equip it with the topology of norm-resolvent convergence. In what follows, the characteristic function on an interval $[a,b]\subset\mathbb{R}$ is denoted $\chi_{[a,b]}$.

\begin{propdef}[\cite{P96,BLP05}]\label{defn:spectral.flow}
 Let $F:I=[0,2\pi]\rightarrow\mathscr{F}^{\rm sa}$ be a continuous path. There exists a partition of $I$,
$
		0=\theta_0<\theta_1<\ldots<\theta_N=2\pi,
$
	and corresponding $a_i>0, i=1,\ldots,N$, such that each subpath of spectral projections
	\begin{equation*}
		[\theta_{i-1},\theta_i]\ni\theta\mapsto \chi_{[-a_i,a_i]}(F(\theta))
	\end{equation*}
	is finite-rank and (norm-)continuous.
	For $\theta\in [\theta_{i-1},\theta_i]$, write $r^+_i(\theta):={\rm dim}\,\chi_{[0,a_i]}(F(\theta))$. The spectral flow $F$ (across 0 energy) is defined to be
		\begin{equation}
			{\rm Sf}(F)\equiv {\rm Sf}(\{F(\theta)\}_{\theta\in I}):=\sum_{i=1}^N r^+_i(\theta_i)-r^+_i(\theta_{i-1}).\label{eq:sf.definition}
		\end{equation}	
	 This integer is independent of the choices of partition and $a_i$, and is a homotopy invariant which is additive under concatenation of paths.
\end{propdef}

\begin{figure}[htb]
	\centering
	\includegraphics[scale=1.1]{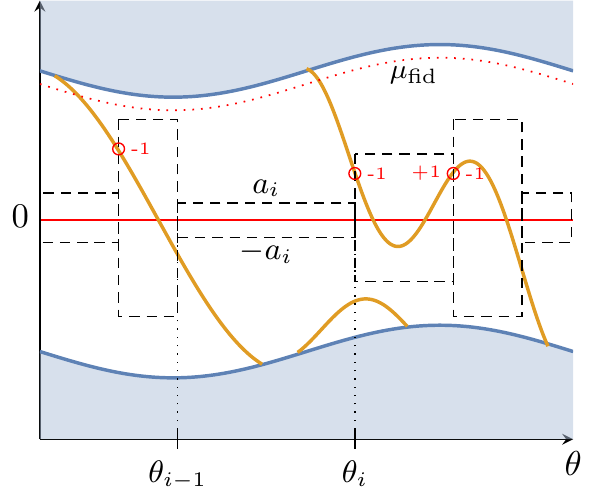}
	\caption{The spectral flow across 0 energy is defined using any jagged cylinder (black dotted lines) of sub-intervals with no leaking of branches on the horizontal edges. It extends Definition~\ref{defn:edge.index.naive} to {situations where the edge mode branches may behave pathologically.}  
		The same spectral flow is obtained across a fiducial curve (red dotted curve) lying slightly below the positive essential spectrum.
\label{fig:spectral_flow}}
\end{figure}

{The above definition of spectral flow is illustrated in Figure~\ref{fig:spectral_flow}.} Along the $i$-th subpath, continuity of $\chi_{[-a_i,a_i]}(F(\theta))$ ensures constancy of the rank, so that no ``leakage'' of eigenvalues occurs across the local upper/lower energy bounds $\pm a_i$. Then the spectral flow across 0 during this subpath is simply the net increase in the number of non-negative eigenvalues at the left endpoint $\theta_{i-1}$ as compared to that at the right endpoint $\theta_i$. Summing up these local spectral flows gives the spectral flow along the full path.

Notice that any eigenvalues above/below $\pm a_i$ are, by construction, irrelevant to the local spectral flows, so it does not matter that they might accumulate into the essential spectrum. Also, the transversality of eigenvalue curves across 0 is not required.

\subsection{Spectral flow across an energy curve}

In Eq.\ \eqref{eq:fiducial.versus.straight}, we wrote that in the unperturbed model, the edge mode crossings can be taken with respect to a constant energy $\mu$ in the bulk gap, or a fiducial energy curve $\mu_{\rm fid}$ slightly below/above the upper/flat bulk band. For the latter, a notion of emergence/disappearance from a bulk band is used (Definition \ref{defn:edge.index}). For a perturbed model, these counts may become ill-defined, but their replacements by spectral flows a la Phillips are always well-defined.

Usually, for a path $F:[0,2\pi]\rightarrow\mathscr{F}^{\rm sa}$, the integer ${\rm Sf}(F)$ measures the net eigenvalue flow across the 0-energy reference curve. {More generally,} we may wish to adopt another continuous fiducial reference curve $\mu_{\rm fid}:[0,2\pi]\rightarrow \mathbb{R}$ such that
\begin{equation*}
 \mu_{\rm fid}(\theta)\not\in\sigma_{\rm ess}(F(\theta)),\;\;\forall\theta\in [0,2\pi],
\end{equation*}
and consider the eigenvalue flow across $\mu_{\rm fid}$, namely, the quantity ${\rm Sf}(F-\mu_{\rm fid})$.

The original path $F$ may be considered as being homotopic to a concatenation of three continuous paths: (i) turning on $\mu_{\rm fid}$, (ii) $F-\mu_{\rm fid}$, (iii) turning off $\mu_{\rm fid}$.
Spectral flow is homotopy invariant and additive under concatenation of paths. This means that we have 
\begin{align}
{\rm Sf}(F)&={\rm Sf}\left(\{F(0)-t\mu_{\rm fid}(0)\}_{t\in [0,1]}\right)+{\rm Sf}\left(F-\mu_{\rm fid}\right)+{\rm Sf}\left(\{F(2\pi)-\mu_{\rm fid}(2\pi)+t\mu_{\rm fid}(2\pi)\}_{t\in [0,1]}\right)\nonumber\\
&=-\#_{[0,\mu_{\rm fid}(0))}F(0) + {\rm Sf}\left(F-\mu_{\rm fid}\right) + \#_{[0,\mu_{\rm fid}(2\pi))}F(2\pi).\label{eq:sf.stokes}
\end{align}
Here, $\#_{[0,\lambda)}$ denotes the number of eigenvalues lying in the energy interval $[0,\lambda)$, counted with multiplicity. From Eq. \eqref{eq:sf.stokes}, we deduce that \emph{for an operator loop}, $F(0)=F(2\pi)$, \emph{the spectral flow is independent of the choice of reference energy loop} (with $\mu_{\rm fid}(0)=\mu_{\rm fid}(2\pi)$),
\begin{equation}
{\rm Sf}(F)={\rm Sf}(F-\mu_{\rm fid}),\label{eq:fiducial.independence}
\end{equation}
generalizing Eq.\ \eqref{eq:fiducial.versus.straight}.

\subsubsection{Relation to (e)merging events {at bulk-band edges}}
In Definition \ref{defn:spectral.flow} of spectral flow of an operator loop, we use a partition by $\theta_i\in [0,2\pi]$ and local energy bounds $\pm a_i$ for the $i$-th subpath. The union of the rectangles $[\theta_{i-1},\theta_i]\times [-a_i,a_i]$ over $i=1,\ldots N$ is a jagged cylinder with jagged upper (horizontal) boundary across which no leakage of eigenvalues occurs{, see Figure \ref{fig:spectral_flow}}.

Write $\Lambda(\theta)$ for the bottom of the positive essential spectrum {of $F(\theta)$.} 
We pick the fiducial reference energy curve $\mu_{\rm fid}$ to lie slightly below $\Lambda$ and above the jagged boundary (Figure \ref{fig:spectral_flow}). Suppose, in addition, that the discrete spectrum of $F(\theta)$ is finite for all $\theta$ (this is satisfied in our unperturbed model). Since at each $\theta_i, i=0,\ldots, N$, there is a maximal eigenvalue below the positive essential spectrum (no accumulation occurs), we may furthermore arrange for $\mu_{\rm fid}$ to be large enough so that
\begin{equation}
\#_{[\mu_{\rm fid}(\theta_i), \Lambda(\theta_i))}=0,\qquad i=0,\ldots,N.\label{eq:extra.gap.condition}
\end{equation}
Compare, for the $i$-th subpath, the spectral flow across $\mu_{\rm fid}$ with the spectral flow across $a_i$. By construction, the latter vanishes (``no leakage''). We deduce, in the same way as Eq.\ \eqref{eq:sf.stokes}, that
\begin{align}
{\rm Sf}(\{F(\theta)-\mu_{\rm fid}(\theta)\}_{\theta\in[\theta_{i-1},\theta_i]})&=\#_{[a_i,\mu_{\rm fid}(\theta_{i-1}))}F(\theta_{i-1})-\#_{[a_i,\mu_{\rm fid}(\theta_{i}))}F(\theta_{i})\label{eq:local.sf.bound.state}\\
&=\#_{[a_i,\Lambda(\theta_{i-1}))}F(\theta_{i-1})-\#_{[a_i,\Lambda(\theta_{i}))}F(\theta_{i}),\nonumber
\end{align}
with the last equality due to Eq.\ \eqref{eq:extra.gap.condition}. 
Thus the $i$-th local spectral flow is simply the change in the number of bound states lying between $a_i$ and $\Lambda$. {Up to a sign, this is}  Definition~\ref{defn:edge.index} applied to the $i$-th subpath, for which a relative Levinson's theorem like Proposition~\ref{prop:relative.Levinson} applies. The full spectral flow of $F$ is then the sum of such local changes in the bound state count.

Notice that we cannot generally count the bound states starting from a \emph{single} global energy level $a$ for the entire loop --- this would entail a global ``no leakage'' condition across $a$, and therefore no spectral flow at all. The point is that bound states are allowed to disappear into {or emerge from a lower band of essential spectrum}  (a lower-lying bulk band). In contrast, for more usual applications of original Levinson's theorem to bounded-below operators, there does exist some global lower energy bound $a$ which is never breached {by the bound states}.

\begin{rem}
If accumulation of eigenvalues occurs, we have to count the (change in the) number of states lying below a fiducial energy slightly smaller than the essential spectrum, as in Eq.\ \eqref{eq:local.sf.bound.state}.
\end{rem}

\subsection{Spectral flow structure of shallow-water wave model}\label{sec:full.sf.structure}
The scattering theory analysis of \cite{GJT21} and Section \ref{sec:BEC} show that
\begin{equation*}
\sigma_{\rm ess}(H^\sharp(k_x,a))=\sigma(H(k_x)) =\; (-\infty,-\sqrt{k_x^2+(f-\nu k_x^2)^2}]\,\cup\,\{0\}\,\cup\,[\sqrt{k_x^2+(f-\nu k_x^2)^2},\infty).
\end{equation*}
Thus $H^\sharp(k_x,a)$ has a negative essential spectral gap, and a positive essential spectral gap.

\medskip
Theorem \ref{thm:gap.continuity} on the continuity of the parametrization $(k_x,a)\mapsto H^\sharp(k_x,a)$ means that the spectral flow along any path $I\rightarrow \mathring{C}$, across any continuous energy curve $\mu_{\rm fid}$ in the positive essential spectral gap, is well-defined. In the case of a loop in $\mathring{C}$, the choice of reference energy loop $\mu_{\rm fid}$ is immaterial (Eq.\ \eqref{eq:fiducial.independence}), and the following definition is meaningful.
\begin{defn}\label{defn:spec.flow.structure}
Let $\ell:I\rightarrow \mathring{C}$ be any continuous loop. {We define
\begin{equation*}
{\rm Sf}^+(\ell):={\rm Sf}\left(\left\{(H^\sharp-\mu_{\rm fid})(\ell(\theta))\right\}_{\theta\in I}\right),
\end{equation*}
where $\mu_{\rm fid}\equiv\mu_{\rm fid}(\cdot)$ is any continuous fiducial energy function valued in the positive essential spectral gap of $H^\sharp(\cdot)$.} 
This assignment of integers ${\rm Sf}^+(\ell)$ to each loop $\ell$ in $\mathring{C}$ is called the \emph{spectral flow structure} of the shallow-water wave model on the half-plane.
\end{defn}

The spectral flow only depends on the homotopy class of $\ell$, so ${\rm Sf}^+$ descends to a group homomorphism 
\begin{equation*}
{\rm Sf}^+(\pi_1(\mathring{C}))\rightarrow\ZZ,
\end{equation*}
and it is enough to work this out for a set of generators of $\pi_1(\mathring{C})$. It is easy to see that $\pi_1(\mathring{C})\cong F_2$, the free group on two generators, since $\mathring{C}$ is homotopy equivalent to a wedge of two circles. Thus the generating loops in $\mathring{C}$ can be taken to be a loop $\ell_+$ at any fixed $k_x>0$, together with a loop $\ell_-$ at any fixed $k_x^\prime<0$, oriented along increasing $a$ (Figure \ref{fig:cylinder}).

There now appears to be two independent edge topological invariants, ${\rm Sf}^+(\ell_+)$ and ${\rm Sf}^+(\ell_-)$, whereas there is only one bulk Chern number. Actually, we have the following relationship:
\begin{prop}\label{prop:right.left.sf}
Let $\ell_\pm:I\rightarrow \mathring{C}$ be the loop at some constant $k_x\gtrless 0$, parametrized by increasing $a\in\mathbb{R}\cup\{\infty\}$. Then ${\rm Sf}^+(\ell_+)=-{\rm Sf}^+(\ell_-)$.
\end{prop}
\begin{proof}
For $0\leq t\leq 1$, consider the self-adjoint operators
\begin{equation*}
H^\sharp(k_x,a;t):=\begin{pmatrix} 0 & tk_x & -\ii\frac{d}{dy} \\ tk_x & 0 & -\ii(f-\nu(k_x^2-\frac{d^2}{dy^2})) \\ -\ii\frac{d}{dy} & \ii(f-\nu(k_x^2-\frac{d^2}{dy^2})) & 0\end{pmatrix},\qquad {\rm subject\;to}\;\eqref{eq:boundary_condition},
\end{equation*}
which at $t=1$ are just the half-line operators $H^\sharp(k_x,a)$ that we have been studying. Fix some large $\tilde{k}_x>\sqrt{f/\nu}$, then the positive essential spectral gap of $H^\sharp(\pm \tilde{k}_x,a;t)$ is $(0,\sqrt{t^2\tilde{k}_x^2+(f-\nu \tilde{k}_x^2)^2})$ and remains open as $t$ is decreased from $1$ to $0$. We shall take $\mu_{\rm fid}$ to be some small constant positive number $\mu$ (say $\mu=f/2$) which stays in this gap. We can compute ${\rm Sf}^+(\ell_\pm)$ as the spectral flow along the loop at $\pm \tilde{k}_x$ (due to homotopy invariance),
\begin{equation}
{\rm Sf}^+(\ell_\pm)={\rm Sf}\left(a\mapsto (H^\sharp(\pm\tilde{k}_x,a;t=1)-\mu)\right)\overset{t\rightarrow 0}{=}{\rm Sf}\left(a\mapsto (H^\sharp(\pm\tilde{k}_x,a;t=0)-\mu)\right).\label{eqn:homotoped.sf}
\end{equation}
Now observe that $H^\sharp(\tilde{k}_x,a;0)$ and $H^\sharp(-\tilde{k}_x,a;0)$ are the same operators, except that the former is subject to boundary condition $a/\tilde{k}_x$ whereas the latter is subject to $a/(-\tilde{k}_x)=(-a)/\tilde{k}_x$, i.e., the loop parameter $a$ is swapped for $-a$. Therefore, the right side of Eq. \eqref{eqn:homotoped.sf} acquires a sign change when $\ell_+$ is changed to $\ell_-$, giving the desired result.
\end{proof}

{The proof of Proposition~\ref{cor:bc.pumping} is now straightforward. }
The anticlockwise loop $\ell_0$ encircling the singularity $(0,0)$ is (homotopy equivalent to) the concatenation of $\ell_+$ and $\ell_-^{\rm op}$, where $(\cdot)^{\rm op}$ denotes orientation reversal. Thus
\begin{equation*}
{\rm Sf}^+(\ell_0)= {\rm Sf}^+(\ell_+) + {\rm Sf}^+(\ell_-^{\rm op})={\rm Sf}^+(\ell_+) - {\rm Sf}^+(\ell_-)=2\,{\rm Sf}^+(\ell_+),
\end{equation*}
where the last equality follows from Prop.\ \ref{prop:right.left.sf}. In Corollary \ref{cor:perturbed}, we saw that the spectral flow around the singularity is minus the Chern number of the upper bulk band,
${\rm Sf}^+(\ell_0)=-C_+=-2$.

\begin{rem}\label{rem:negative.sf}
We can similarly define ${\rm Sf}^-$ as the spectral flow structure across energy curves in the lower essential spectral gap. The lower bulk band has Chern number $C_-=-2$, and the bound state counting above this band via Levinson's theorem needs a sign change. Overall, the conclusion is that ${\rm Sf}^-={\rm Sf}^+$.
\end{rem}

\begin{rem}
 In the {3D} Weyl semimetal setting from \cite{Thiang21}, one encounters a singularity in the {2D} boundary momentum space parameter due to a closing of the essential spectral gap at the projected \emph{Weyl point}, leading to a loss of the Fredholm condition there. This singularity does not involve the boundary condition, so the \emph{Fermi arcs} as determined by the spectral flows are present whatever the boundary condition. A peculiar feature of the shallow-water model is that the parameter space singularity is not due to gap closing, but rather a loss of self-adjointness. 
\end{rem}

\end{document}